\documentclass[journal]{IEEEtran}

\usepackage{amsmath}
\usepackage{amssymb}
\usepackage{bigfoot}

\newtheorem{theorem}{Theorem}
\newtheorem{lemma}{Lemma}

\usepackage{graphicx}
\graphicspath{{generated/figs/}{generated/cases/}}

\usepackage{booktabs}
\usepackage{multirow}
\usepackage{siunitx}

\usepackage{xcolor}
\usepackage{dblfloatfix}
\newcommand\blfootnote[1]{%
  \begingroup
  \renewcommand\thefootnote{}\footnote{#1}%
  \addtocounter{footnote}{-1}%
  \endgroup
}
\usepackage{subcaption}
\usepackage{placeins}
\usepackage{cuted}
\usepackage{enumitem}

\usepackage{algorithm}
\usepackage{algorithmic}

\usepackage{url}
\usepackage{cite}

\begin{document}

\title{FABO: Agent-Guided Discovery of Joint Breakpoint Optimization for Timing-Driven Routing Trees}

% --- PLACEHOLDER author block: fill in real names/affiliations before submission ---
\author{Shang Liu, Wenji Fang,~\IEEEmembership{Graduate Student Member,~IEEE}, Jing Wang, Hongxin Kong, Yao Lu, \\ Zhiyao Xie,~\IEEEmembership{Member,~IEEE}\vspace{-.2in}}

\maketitle

% --- Abstract ---
\begin{abstract}
The topology of a routing tree determines how a multi-pin net branches and shares physical wire, directly affecting wirelength, congestion, capacitance, and delay.
We study a central early-stage routing problem: minimizing wirelength while bounding the root-to-sink path stretch for every sink.
SALT is the state-of-the-art constructive algorithm for this problem.
We ask whether language-model-guided search can discover a constructive algorithm that improves on SALT.
To make this search reliable, we develop an agent framework that combines parallel exploration with independent checking.
Applied to SALT, the framework discovers a structural limitation: SALT repairs one sink path at a time and therefore never jointly decides where paths sharing root-side wire should split.
This sink-local choice can split the paths too early and duplicate wire.
This discovery leads to \textbf{Flow-Aware Breakpoint Optimization (FABO)}, which jointly optimizes breakpoints across root-to-sink paths that share wire while preserving every sink's stretch budget.
Across 1.29 million ICCAD15 nets and SALT's 20-point stretch-tolerance schedule, FABO reduces average FLUTE-normalized wirelength at every setting, with peak same-$\epsilon$ reductions of $0.83\%$ overall and $2.66\%$ for nets with at least 30 pins.
With $1.3\times$ SALT's runtime, FABO-FAST identifies and optimizes most nets for which FABO provides a substantial wirelength reduction.
Code is available at \url{https://github.com/DevinShang/routing-FABO}.
\end{abstract}


\blfootnote{\quad\ Shang Liu, Wenji Fang, Jing Wang, Yao Lu and Zhiyao Xie are with the Department of Electronic and Computer Engineering at the Hong Kong University of Science and Technology, Hong Kong SAR, China. {Corresponding Author: Zhiyao Xie (eezhiyao@ust.hk).}}
\blfootnote{\quad\  Hongxin Kong is with Synopsys, United States.}
% --- Section bodies (verbatim from acmart, order preserved) ---
\section{Introduction}

Before detailed routing, a routing tree decides where the root-to-sink paths of a multi-pin net branch and which paths share physical wire.
These decisions directly affect total wirelength, routing demand, interconnect capacitance, congestion, and root-to-sink delay~\cite{alpert2006practicallyfree,chu2008flute}.
The effect is strongest on high-fanout nets: one root-side segment may serve many sinks, so splitting those paths early can repeat substantial wire across the resulting branches.

\begin{figure}[t]
  \centering
  \includegraphics[width=\columnwidth]{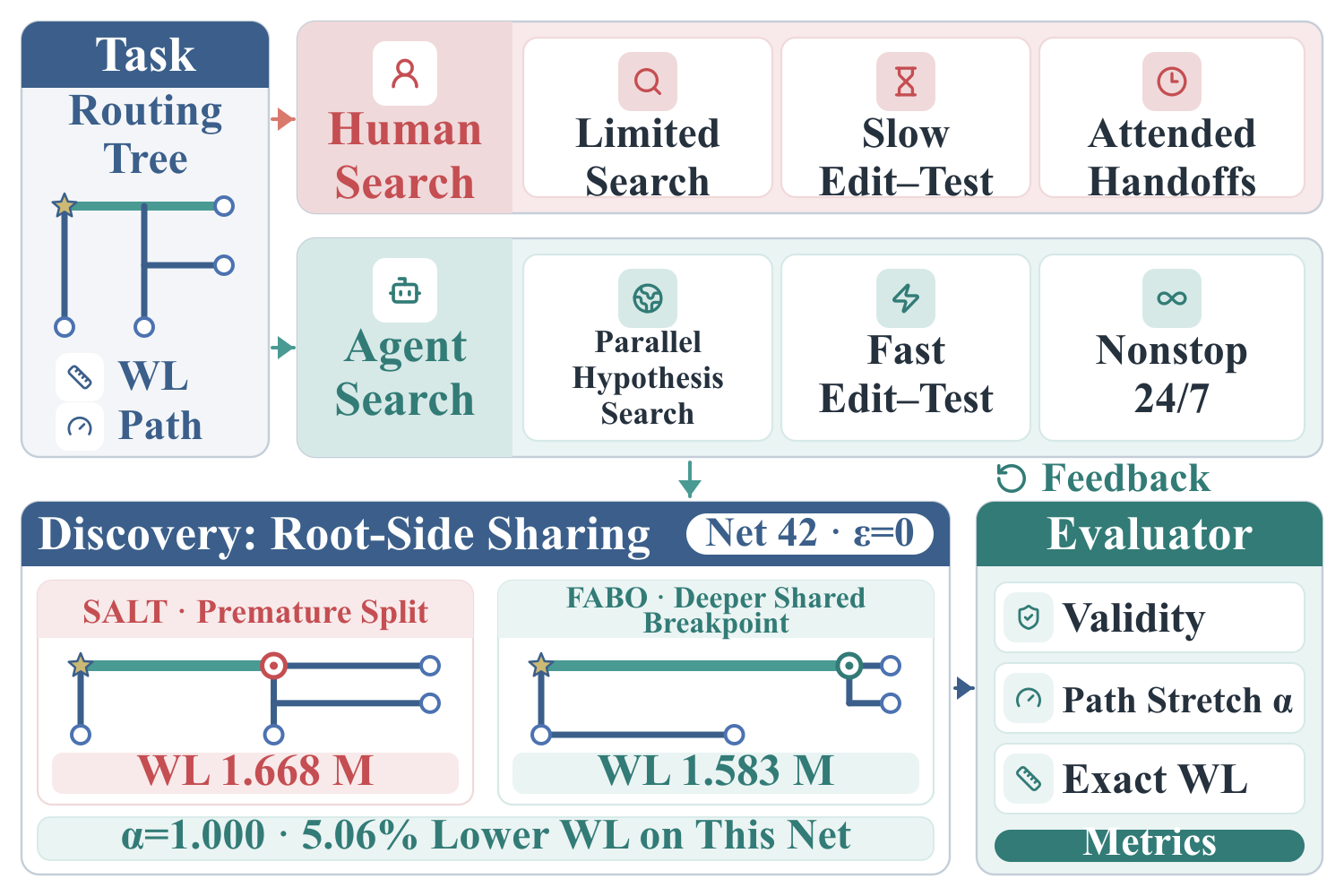}
  \caption{\textbf{Agent-guided discovery and its routing outcome.}
  The upper panel contrasts the search capabilities of the human and agent workflows.
  The lower panel shows the recorded superblue16 net 42 case at $\epsilon=0$: SALT and FABO both achieve $\alpha=1.000$, while FABO moves the shared breakpoint deeper and reduces wirelength by $5.06\%$.}
  \label{fig:intro_contribution}
\end{figure}

FLUTE rapidly constructs near-minimum-wirelength rectilinear Steiner trees but does not bound the path detour of each sink~\cite{flute_iccad04,chu2008flute}.
SALT~\cite{salt_iccad17} instead minimizes wirelength subject to a hard normalized-stretch bound for every sink.
This is a fundamental early-design formulation: its geometry-only constraint applies before detailed resistance, capacitance, or sign-off timing is available, yet it protects every sink individually.
SALT is the state-of-the-art constructive algorithm for this problem.
The key question is whether a new algorithm can reduce wirelength beyond SALT under the same per-sink budgets, giving later physical-design stages a lighter routing structure without weakening any path guarantee.

PD-Rev balances wirelength and path detour in the Prim--Dijkstra family~\cite{pdrev_ispd18,alpert1995primdijkstra}; PatLabor~\cite{chen2025patlabor} constructs Pareto sets for wirelength and maximum path length; and Wu et al.~\cite{wu2025delaydriven} optimize Elmore delay under a wirelength requirement.
These objectives are complementary, but none is a direct replacement for SALT under the per-sink normalized-stretch constraint in Section~\ref{sec:problem_setup}.
% Our central question is: can a new constructive algorithm beat SALT on the problem that SALT defines? 

\textbf{From agent search to a new algorithm.}
Recent program-search systems show that language-model-guided search can discover new mathematical and heuristic programs~\cite{funsearch_nature24,liu2024eoh,novikov2025alphaevolve,yu_npcomplete_2025}.
Existing EDA agents mainly orchestrate existing tools~\cite{he2023chateda,ghose2025orfsagent,pasandi2025jarvis}, while algorithm-configuration methods tune parameters of a fixed procedure~\cite{paramils_jair09,smac_lion11,irace_orp16}.
We instead build an agent-guided research system to develop and validate a new routing-tree construction algorithm.

A serial agent has two weaknesses in this NP-hard search.
It follows one reasoning trajectory, so after adopting a construction it tends to patch that construction instead of exploring a fundamentally different one.
It must also carry each idea through specification, implementation, verification, and evaluation; an unchecked error at any stage then shapes every later decision.
Our framework addresses both weaknesses through multi-agent generation and checking.
Eight isolated Hypothesis Agents explore different structural mechanisms, the Critic checks their assumptions and converts the strongest directions into fixed specifications, and three Programmer--Verifier chains implement and independently verify those specifications in parallel.
The fixed Evaluator admits results from valid trees into the SSOT for the next round.

\textbf{Flow-Aware Breakpoint Optimization (FABO).}
Per-sink budgets constrain root-to-sink path lengths; the objective minimizes total wirelength.
A longer feasible path can reduce wirelength when it lets several sinks share more wire.
A shared segment disappears from the tree only after every path using it is rerouted, so this saving must be evaluated jointly across those paths.
SALT builds a light support tree and repairs each sink whose root path exceeds its budget: it opens a \emph{breakpoint} on that path, keeps the wire from the breakpoint to the sink, and reconnects the breakpoint to the root through a shortest wire.
The repair is greedy and per sink, and the postprocessing that follows moves one node at a time; sharing recovers only as a side effect, and the joint choice of where a group of paths should stop sharing one trunk is never represented or searched.
The unexploited sharing grows with fanout, where many paths share one long root-side trunk.

In FABO, a \emph{flow} is one root-to-sink path in the support tree.
Instead of committing a flow to one repair point, FABO represents its breakpoint choices as one \emph{feasible breakpoint interval}.
The interval exactly characterizes all budget-feasible breakpoints.
Moving a breakpoint deeper shortens the kept wire at least as much as it can lengthen the new root wire; therefore, the feasible breakpoints form one contiguous interval ending at the sink.
Where intervals overlap, their flows can adopt one \emph{shared breakpoint} without violating any sink budget, and FABO counts a support edge as released only when every flow using it leaves.
Joint options that sink-by-sink repair never exposes thus become explicit and comparable, while every candidate tree is still checked against every sink budget.

We evaluate FABO on 1.29 million ICCAD15 nets with at least three pins under SALT's 20-point $\epsilon$ schedule.
FABO reduces average FLUTE-normalized wirelength at every setting, with peak same-$\epsilon$ reductions of $0.83\%$ overall and $2.66\%$ among nets with at least 30 pins. Our contributions are:
\begin{itemize}
\item \textbf{A multi-agent framework for reliable algorithm discovery.} Parallel Hypothesis Agents explore different construction families. At each stage, other agents check the hypotheses, specifications, implementations, and evaluation outputs before the evidence enters the next round. 
\item \textbf{Discovery of a structural bottleneck.} SALT repairs one sink path at a time, so it never chooses a breakpoint jointly for paths that share the same root-side wire. Its final tree can therefore split these paths earlier than necessary and duplicate wire.
\item \textbf{Flow-Aware Breakpoint Optimization.} FABO changes the decision unit from one sink to the group of paths sharing root-side wire. It represents each path's legal breakpoints as one interval and turns joint repair into selecting shared points from interval overlaps, making group-level wirelength gains explicit while preserving every sink's path budget.
\item \textbf{Quality at scale with a practical runtime option.} Across 1.29 million nets and all 20 SALT tolerances, FABO reduces average FLUTE-normalized wirelength at every setting, with peak reductions of $0.83\%$ overall and $2.66\%$ for nets with at least 30 pins. With $1.3\times$ SALT's runtime, FABO-FAST identifies and optimizes most nets for which FABO provides a substantial wirelength reduction.
\end{itemize}

\section{Problem Formulation}
\label{sec:problem_setup}

\begin{figure*}[t]
  \centering
  \includegraphics[width=0.9\textwidth]{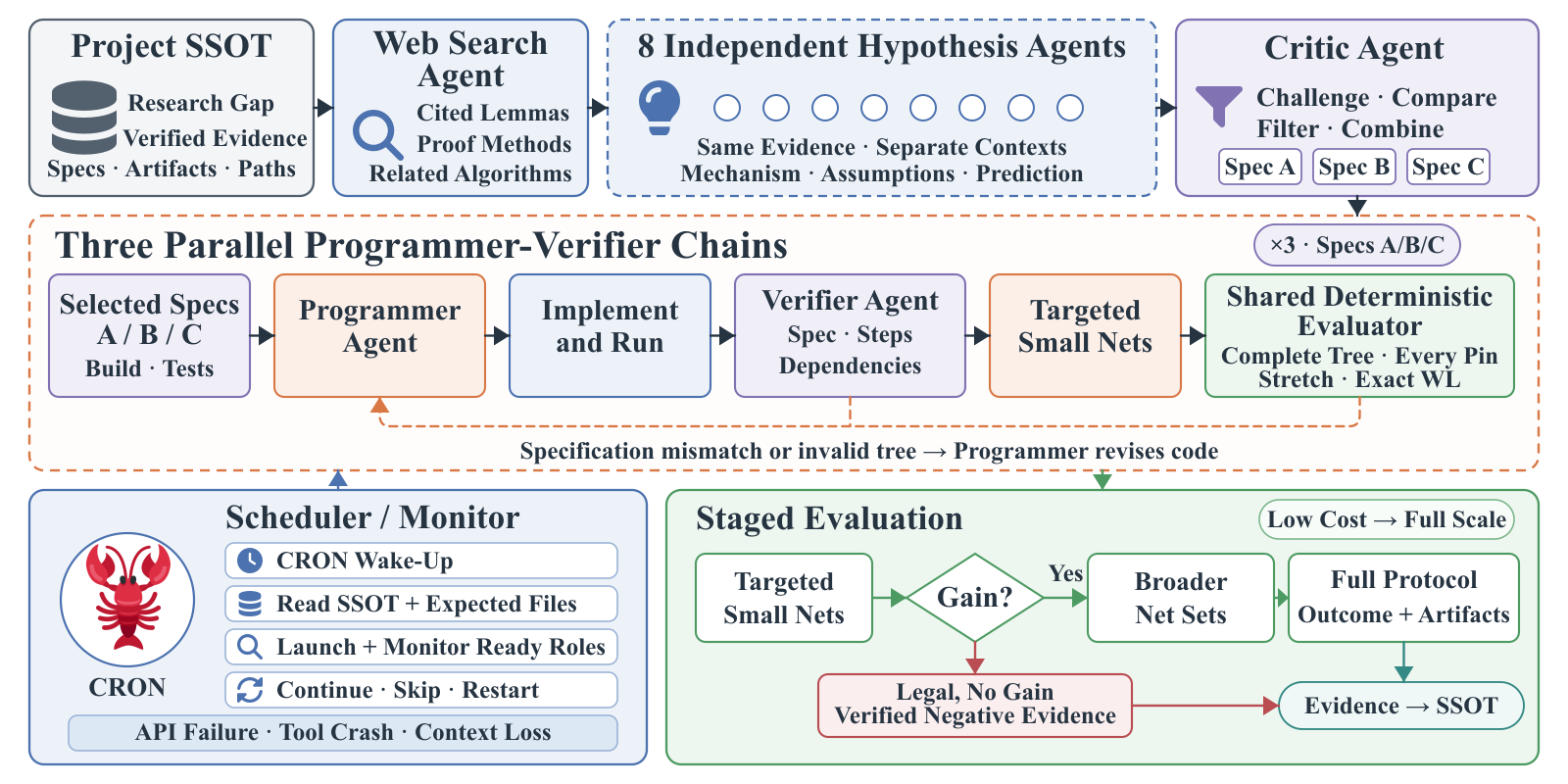}
  \caption{\textbf{Multi-agent generation and checking in each search round.}
  Eight isolated Hypothesis Agents explore different structural mechanisms from the same verified evidence, and the Critic checks their assumptions before producing three fixed specs.
  Three Programmer--Verifier chains implement and independently check the specs in parallel. The Evaluator verifies the resulting trees, the SSOT records the evidence, and the Scheduler/Monitor keeps the file-based workflow running.}
  \label{fig:agent_closed_loop_concept}
\end{figure*}

We adopt the rectilinear Steiner shallow-light formulation of SALT~\cite{salt_iccad17}.
The input consists of a root pin $r$ and a sink set $P=\{t_1,\dots,t_m\}$ on the Manhattan plane.
A feasible solution is a rooted rectilinear Steiner tree $T$ spanning the root pin $r$ and every sink pin $t\in P$.
The pin coordinates are fixed, while the routing tree $T$ may introduce Steiner points.

\textbf{Path distance and shallowness:}
The Manhattan distance $d_M(u,v)$ is the length of a shortest rectilinear path between points $u$ and $v$.
The tree-path distance $d_T(r,t)$ is the length of the unique path from the root pin $r$ to a sink pin $t$ in the routing tree $T$.
The shallowness $\alpha(T)$ is the largest normalized root-to-sink path length:
\begin{equation}
\alpha(T)=\max_{t\in P}\frac{d_T(r,t)}{d_M(r,t)}.
\end{equation}

\textbf{Wirelength and normalization:}
The wirelength $WL(T)$ is the exact length of the geometric union of all rectilinear segments in the routing tree $T$; coincident segments are counted once.
Let $T_{\mathrm{FLUTE}}$ denote the FLUTE tree for the same root pin $r$ and sink set $P$.
Following SALT's evaluator, we report the FLUTE-normalized wirelength
\begin{equation}
\beta(T)=\frac{WL(T)}{WL(T_{\mathrm{FLUTE}})}.
\end{equation}

\textbf{Normalized delay:}
The experiments also report the SALT normalized-delay metric $\gamma(T)$~\cite{salt_iccad17}.
Let $D_T(r,t)$ be the Elmore delay from the root pin $r$ to a sink pin $t$ under the unit-length resistance and capacitance model.
Following SALT's evaluator, let $D_{\mathrm{LB}}(r,t)$ denote SALT's lower-bound delay for sink pin $t$, and define the per-net lower bound as $D_{\mathrm{LB}}(r,P)=\max_{t\in P}D_{\mathrm{LB}}(r,t)$.
Then
\begin{equation}
\gamma(T)=\frac{\max_{t\in P}D_T(r,t)}{D_{\mathrm{LB}}(r,P)}.
\end{equation}
The normalized delay $\gamma(T)$ is a reporting metric rather than an optimization constraint in this work.

\textbf{Stretch-constrained objective:}
Let the stretch tolerance $\epsilon\ge 0$ specify the allowed multiplicative detour beyond a Manhattan shortest path.
For each sink pin $t$, define the per-sink path budget $B_t$ by
\begin{equation}
B_t=(1+\epsilon)d_M(r,t).
\label{eq:path_budget}
\end{equation}
The rectilinear shallow-light routing problem is
\begin{equation}
\min_T \; WL(T)
\quad
\text{s.t.}
\quad
d_T(r,t)\le B_t, \quad \forall t\in P.
\label{eq:salt_formulation}
\end{equation}
Because $WL(T_{\mathrm{FLUTE}})$ is fixed for a given net, minimizing $WL(T)$ is equivalent to minimizing $\beta(T)$.
At stretch tolerance $\epsilon=0$, every root-to-sink path must be Manhattan-shortest, giving the rectilinear minimum-arborescence limit.
As the stretch tolerance $\epsilon$ increases, the path constraint relaxes and the minimum-wirelength feasible tree approaches a rectilinear Steiner minimum tree~\cite{salt_iccad17}.

% =====================================================
% 3.1  THE ALGORITHM-DISCOVERY FRAMEWORK
% =====================================================
\section{Agent-Guided Algorithm-Discovery Framework}
\label{sec:discovery_framework}

\textbf{Why a serial single-agent workflow is fragile.}
A serial agent explores one reasoning trajectory.
Once it commits to a construction, later iterations tend to refine that construction instead of proposing a fundamentally different routing structure.
This path dependence can trap the search in a local optimum.
Algorithm discovery also spans several dependent stages: hypothesis generation, specification, implementation, verification, and evaluation.
The output of each stage becomes the input of the next.
A weak hypothesis, incomplete specification, coding error, or invalid evaluation therefore distorts every later decision.

\textbf{Core design: generate alternatives and check every stage.}
Our framework uses multi-agent collaboration throughout the workflow.
At the hypothesis stage, eight agents explore different structural mechanisms and a Critic challenges their assumptions before selecting three specifications.
At the execution stage, three Programmer--Verifier chains implement and independently check the specifications in parallel.
A fixed Evaluator then checks complete routing trees, and the SSOT stores the verified evidence for the next round.
This repeated generate--check pattern broadens the structural search and improves the quality of every handoff.

\textbf{One search round.}
Each round starts from the single source of truth (SSOT), which records verified conclusions and links to their code and evaluation artifacts.
The Web Search Agent retrieves mathematical tools for the unresolved structural question.
Eight Hypothesis Agents receive the same SSOT and references in isolated contexts and propose alternative mechanisms with explicit assumptions and measurable predictions.
The Critic checks the proposals against the routing constraints and existing evidence, then converts the strongest directions into three New Idea Specs.
Three Programmer--Verifier chains implement and check these specs in parallel.
The fixed Evaluator verifies each complete routing tree and every path budget before measuring wirelength.
The SSOT records the verified result and its evidence before the next round starts.

\begin{figure*}[t]
  \centering
  \includegraphics[width=\textwidth]{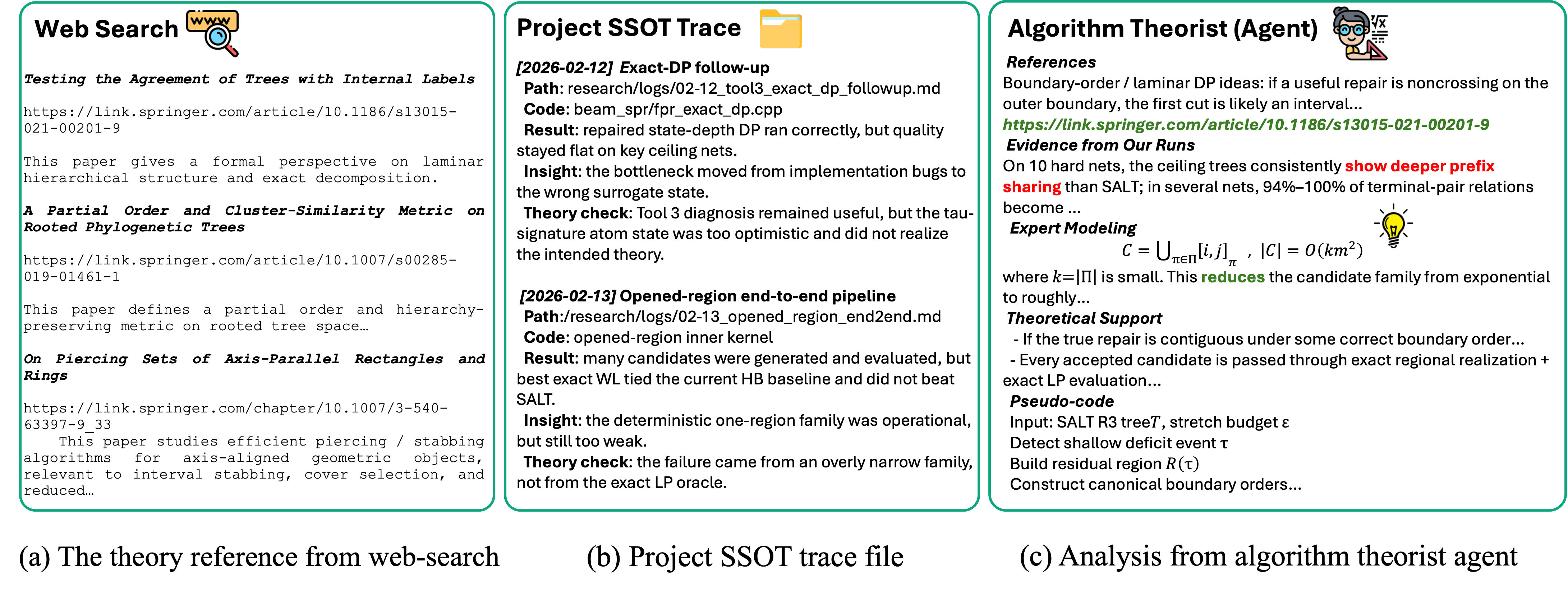}
  \caption{\textbf{Evidence flow through one search round.}
  (a) The Web Search Agent retrieves cited tools on rooted-tree structure and geometric piercing~\cite{fernandezbaca2021agreement,hendriksen2020partial,segal1997piercing}.
  (b) The SSOT records two candidate evaluations with their code paths, measured outcomes, extracted lessons, and theory diagnoses.
  (c) A Hypothesis Agent combines the retrieved tools with the recorded evidence and writes a model, assumptions, theoretical support, and pseudocode.
  The Critic converts this material into a fixed spec for implementation and verification.}
  \label{fig:agent_content_demo}
\end{figure*}

\subsection{Multi-Agent Generation and Checking}
\label{sec:framework_roles}

\textbf{Structural exploration.}
The SSOT defines the unresolved structural question, and the Web Search Agent retrieves relevant lemmas, proof techniques, and algorithms.
Eight Hypothesis Agents independently combine this technical note with the same verified evidence.
Each agent proposes a structural mechanism, states its assumptions, and gives a measurable prediction.
The isolated contexts produce several search directions instead of eight refinements of one shared draft.
The Critic then checks every assumption against the problem constraints and evidence, rejects inconsistent proposals, completes missing construction steps, and combines compatible mechanisms.
It returns three self-contained New Idea Specs with fixed constructions and acceptance tests.

\textbf{Parallel implementation and verification.}
Three Programmer--Verifier chains execute the specs concurrently.
In each chain, the Programmer implements one spec and the Verifier independently checks the implementation, dependencies, and required tests against that spec.
This check prevents implementation errors from being interpreted as algorithmic evidence.
For example, in one SALT reproduction~\cite{salt_iccad17}, empty stubs allowed the binary to run after a missing dependency disabled required postprocessing.
The Verifier classified the run as an implementation failure and returned it for repair.

\textbf{Evaluation and evidence update.}
Each verified implementation first runs on targeted small nets, then on broader net sets, and finally under the full evaluation protocol.
The fixed Evaluator checks the complete output tree, verifies every root-to-sink path budget, and measures exact wirelength.
Thus the framework separately establishes that the code matches the spec, the output is valid, and the valid output improves the objective.
The SSOT accepts the result after these checks and supplies it to the next search round.
\subsection{How Results Update the Next Round}
\label{sec:framework_cycle}
\label{sec:framework_evidence}

\textbf{Result classification.}
The framework assigns every run one of three labels.
An \emph{implementation failure} means that the code or its dependencies differ from the spec.
A \emph{construction failure} means that a spec-conformant implementation produces an invalid routing tree.
A \emph{verified negative result} means that the implementation matches the spec and produces a valid tree, but fails to reduce wirelength.
Algorithm comparison uses evaluator measurements from valid trees.

\textbf{Evidence record.}
Each run produces a code snapshot, build log, output tree, metric file, and verification report.
The SSOT entry stores the result label, a short evidence statement, and paths to these artifacts.
The next round reads this record as its starting evidence.

\textbf{From failed candidates to FABO.}
Figure~\ref{fig:agent_content_demo}(b) shows two implemented candidate families that produced no gain on the targeted nets.
Their diagnoses shifted the next round from local attachment changes to the placement of shared root-side structure.
Panel (c) turns this question into explicit assumptions and pseudocode.
This sequence led to the flow-level breakpoint formulation in Section~\ref{sec:fabo}.

\subsection{Persistent State and Recovery}
\label{sec:framework_recovery}

\textbf{File-based execution.}
Each agent reads its inputs from files and writes its outputs to files.
An external Scheduler/Monitor polls these files at fixed intervals.
It launches an agent when its inputs are ready, restarts an interrupted agent from its latest artifacts, and tracks the three execution chains separately.
\begin{figure*}[t]
\centering
\includegraphics[width=\textwidth]{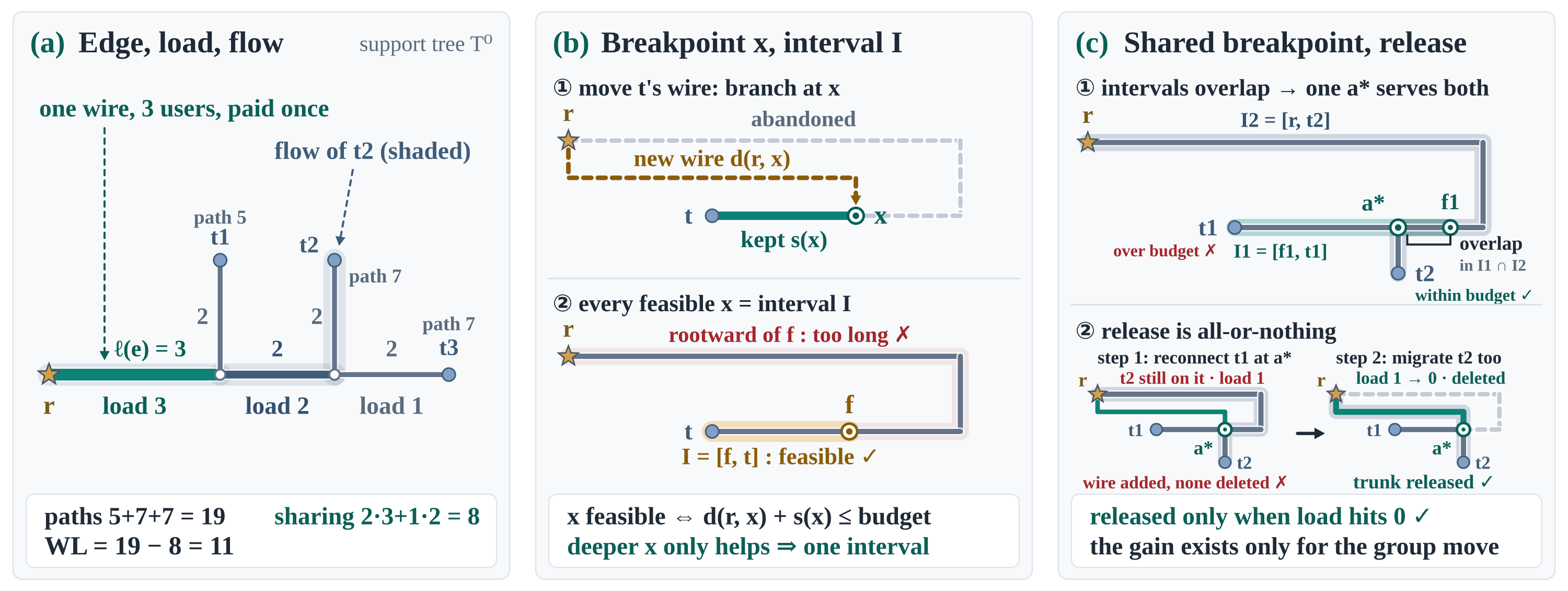}
\caption{
(a) An edge is paid once, no matter how many flows use it: wirelength is total path length minus sharing, Eq.~\eqref{eq:fabo_wl_identity}.
(b) A breakpoint $x$ keeps the wire from $x$ to the sink and adds a new shortest wire from the root to $x$; all feasible $x$ form the interval $I=[f,t]$.
(c) One point $a^\ast$ in the overlap of two intervals reroutes both flows; the abandoned trunk is deleted (dashed) only in step 2, when its load reaches zero.}
\label{fig:fabo_concepts}
\end{figure*}

\begin{figure*}[t]
\centering
\includegraphics[width=\textwidth]{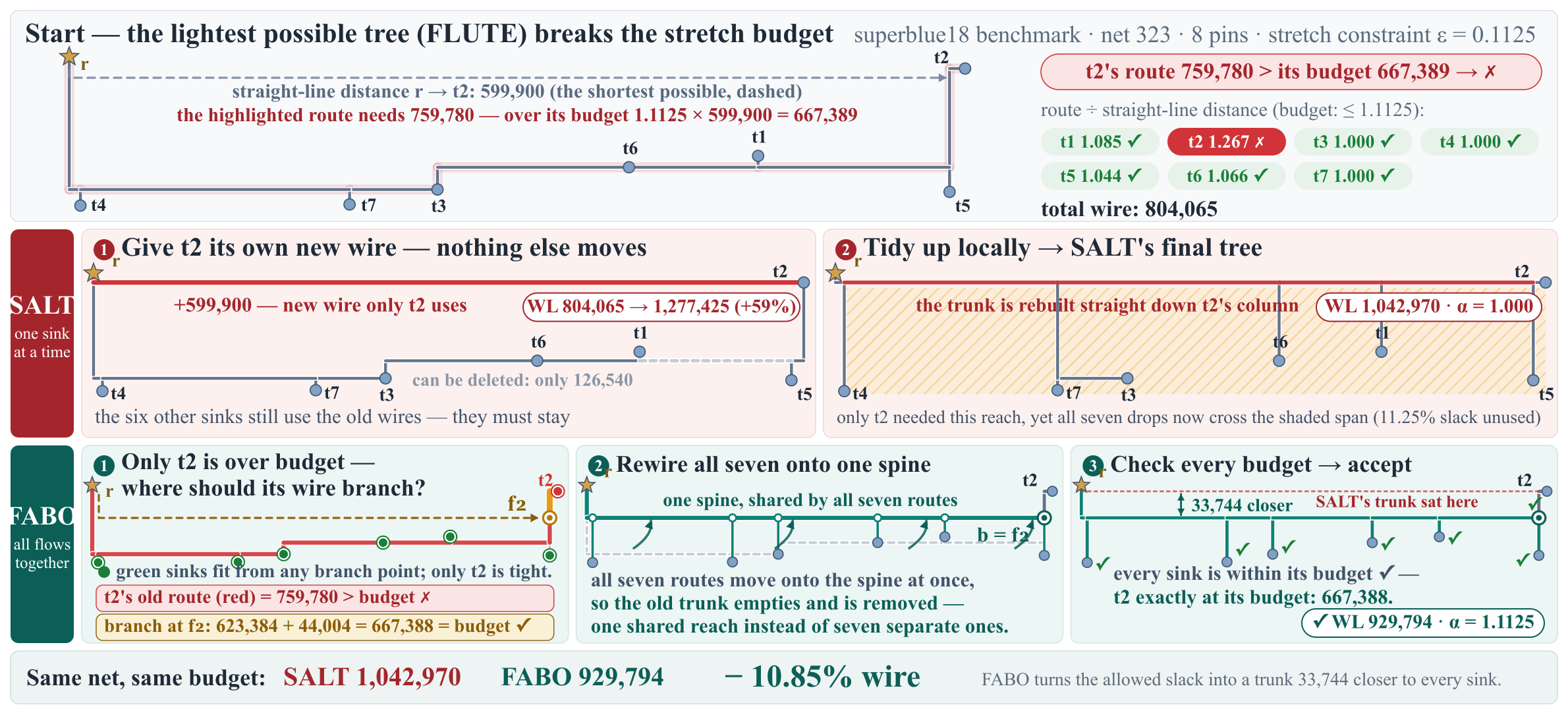}
\caption{\textbf{One recorded net, two repairs (superblue18 net 323, eight pins, $\epsilon=0.1125$).}
SALT gives the only violating sink $t_2$ a private wire and then rebuilds the trunk on $t_2$'s column, using none of the allowed $11.25\%$ stretch.
FABO instead places the shared trunk at $f_2$, rewires all seven flows onto it, and verifies every budget.
Under the same budgets, FABO reduces wirelength from $1{,}042{,}970$ to $929{,}794$ ($10.85\%$).}
\label{fig:algo_demo}
\end{figure*}

\section{FABO: Flow-Aware Breakpoint Optimization}
\label{sec:fabo}

\textbf{Core idea.}
When a sink's path exceeds its budget, SALT builds a new wire for that one sink and leaves every other path unchanged~\cite{salt_iccad17}.
The total wirelength, however, is decided by how much wire the paths \emph{share}.
FABO therefore plans the new wires so that several paths can use them together: for every sink it computes every legal position for the new wire's branch (always one contiguous interval); it lets one branch serve several sinks wherever their intervals overlap.
Figure~\ref{fig:fabo_concepts} shows some important definitions on a toy net.

\textbf{Wirelength is paths minus sharing (Figure~\ref{fig:fabo_concepts}(a)).}
An \emph{edge} $e$ is one wire segment of length $\ell(e)$, and its \emph{load} $\mathrm{load}_T(e)\ge1$ is the number of root-to-sink paths through it.
We call each such path a \emph{flow}: the route one sink currently uses to reach the root.
An edge contributes its length to the wirelength only once, no matter how many flows it carries, so every routing tree satisfies
\begin{equation}
WL(T)=\sum_{t\in P}d_T(r,t)
\;-\;\sum_{e\in T}\bigl(\mathrm{load}_T(e)-1\bigr)\,\ell(e).
\label{eq:fabo_wl_identity}
\end{equation}
In the example, three flows of lengths $5{+}7{+}7=19$ overlap on $8$ units of wire, leaving $WL=11$.
The budget constraint $d_T(r,t)\le B_t$ from Section~\ref{sec:problem_setup} only puts an upper limit on each path length in the first sum; a path may be longer than necessary, as long as it stays under its budget.
FABO searches for reroutings that spend some of the available path-length budget to increase sharing and reduce wirelength: the additional sharing must save more wire than the increase in total root-to-sink path length.

\textbf{Breakpoints and intervals (Figure~\ref{fig:fabo_concepts}(b)).}
When a path goes over budget, SALT connects the detected breakpoint to the root through a new shortest wire. In the running example below, that breakpoint is the sink itself, so the sink receives a private wire~\cite{salt_iccad17}.
FABO turns this fixed repair into a choice: the new wire may stop at any point $x$ of the flow, called a \emph{breakpoint}---the wire from $x$ to the sink is kept (the \emph{suffix}, of length $s(x)$), the wire from the root to $x$ is abandoned, and a new shortest wire of length $d_M(r,x)$ is built from the root to $x$.
The rerouted flow is at most $d_M(r,x)+s(x)$ long, so $x$ is \emph{feasible} when $d_M(r,x)+s(x)\le B_t$.
Moving $x$ toward the sink shortens the suffix by exactly the distance moved and lengthens the new wire by at most that distance, so this bound never increases, and the feasible breakpoints form one interval $I=[f,t]$ ending at the sink (Lemma~\ref{lem:fabo_interval}).
Its \emph{boundary} $f$ is the feasible point closest to the root ($f=r$ when the whole flow is within budget); for an over-budget flow, $f$ lies where the path turns back toward the root, because only there does moving the breakpoint toward the sink shorten both the new root connection and the retained suffix.

\textbf{Shared breakpoints and all-or-nothing release (Figure~\ref{fig:fabo_concepts}(c)).}
A breakpoint $a^\ast$ in the overlap of $I_1$ and $I_2$ is feasible for both flows, allowing them to share one new root wire.
Moving only flow 1 does not release the old trunk because flow 2 still uses it; the trunk is removed only after flow 2 also moves. FABO therefore counts an edge as saved only when all of its users leave---the all-or-nothing, per-edge account that makes the method \emph{flow-aware}.

\textbf{From these definitions to the algorithm.}
FABO computes every flow's interval (Section~\ref{sec:fabo_ranges}); selects a small set of breakpoints so that every interval contains one, sharing wherever intervals overlap (Section~\ref{sec:fabo_selection}); proposes deeper breakpoints ranked by the wire their group releases (Section~\ref{sec:fabo_accounting}); and accepts a candidate only after the complete realized tree passes every budget and improves the incumbent (Section~\ref{sec:fabo_full_tree}).
The joint decision SALT never represents---where a group of flows should stop sharing one trunk---thereby becomes an explicit decision that FABO searches directly.

\textbf{Running example.}
Figure~\ref{fig:algo_demo} traces one recorded net---superblue18 net 323, eight pins, $\epsilon=0.1125$---through both algorithms: the top strip shows the violating support tree, SALT steps~1--2 show the private repair and local cleanup, and FABO steps~1--3 show boundary selection, joint rewiring, and final budget validation.

\subsection{From Sink-Local Repair to Shared Breakpoints}
\label{sec:fabo_motivation}

\textbf{SALT's decision.}
SALT traverses a light support tree and, whenever one sink path exceeds its budget $B_t$, connects the point where the violation is detected---its \emph{breakpoint}---to the root through a new shortest wire; the wire from the breakpoint to the sink is kept, and fixed postprocessing passes then clean the complete tree~\cite{salt_iccad17}.
Every choice in this loop is triggered by, and made for, one sink.

\textbf{SALT's two-step repair (Figure~\ref{fig:algo_demo}, SALT lane).}
In the top strip of Figure~\ref{fig:algo_demo}, the support tree has wirelength $804{,}065$ and exactly one violation: the root-to-sink path of $t_2$ has length $759{,}780$, above its budget $B_{t_2}=1.1125\times599{,}900=667{,}388.75$.
SALT's repair (SALT lane, step~1) gives $t_2$ a private shortest wire of length $599{,}900$.
The other flows still need the old \emph{trunk}---the long root-side run of wire that every flow shares---so only $126{,}540$ of the old wire comes free, and wirelength jumps to $1{,}277{,}425$.
Postprocessing (step~2) then consolidates the tree onto the private wire the repair just added, rebuilding the trunk straight down $t_2$'s column: SALT's final tree has wirelength $1{,}042{,}970$ at maximum stretch $\alpha=1.000$.

\textbf{The missed decision.}
SALT's final tree uses none of the permitted $11.25\%$ stretch: the trunk is pinned to $t_2$'s column, and every \emph{drop}---the wire from the trunk to one sink---crosses the full span between that column and its sink.
The question SALT never asks is where the seven flows, as a group, should stop sharing one trunk.
Reaching the shorter tree requires moving the trunk and every drop together.
Per-sink repair never represents such a change, and SALT's postprocessing does not find it here: its local moves re-attach flows onto existing wire, and the deeper trunk does not yet exist.
FABO asks exactly this question.
In the FABO lane of Figure~\ref{fig:algo_demo}, step~2 rewires all seven flows onto a trunk $33{,}744$ closer to the six other sinks; $t_2$ spends its permitted detour, while no other path gets longer. Step~3 verifies every budget and reports wirelength $929{,}794$---a $10.85\%$ reduction under the same constraints.
In the terms of~\eqref{eq:fabo_wl_identity}, FABO spends $67{,}488$ more total path length---all of it $t_2$'s permitted detour---to save $180{,}664$ more through sharing; the difference is exactly the $113{,}176$ of wire removed.
Figure~\ref{fig:fabo_high_gain_cases} shows the same premature root-side splitting on eight high-degree nets; Section~\ref{sec:results} shows that FABO saves more wire on higher-fanout nets and on nets where SALT shares less root-side wire.

\begin{figure*}[!t]
\centering
\includegraphics[width=\textwidth]{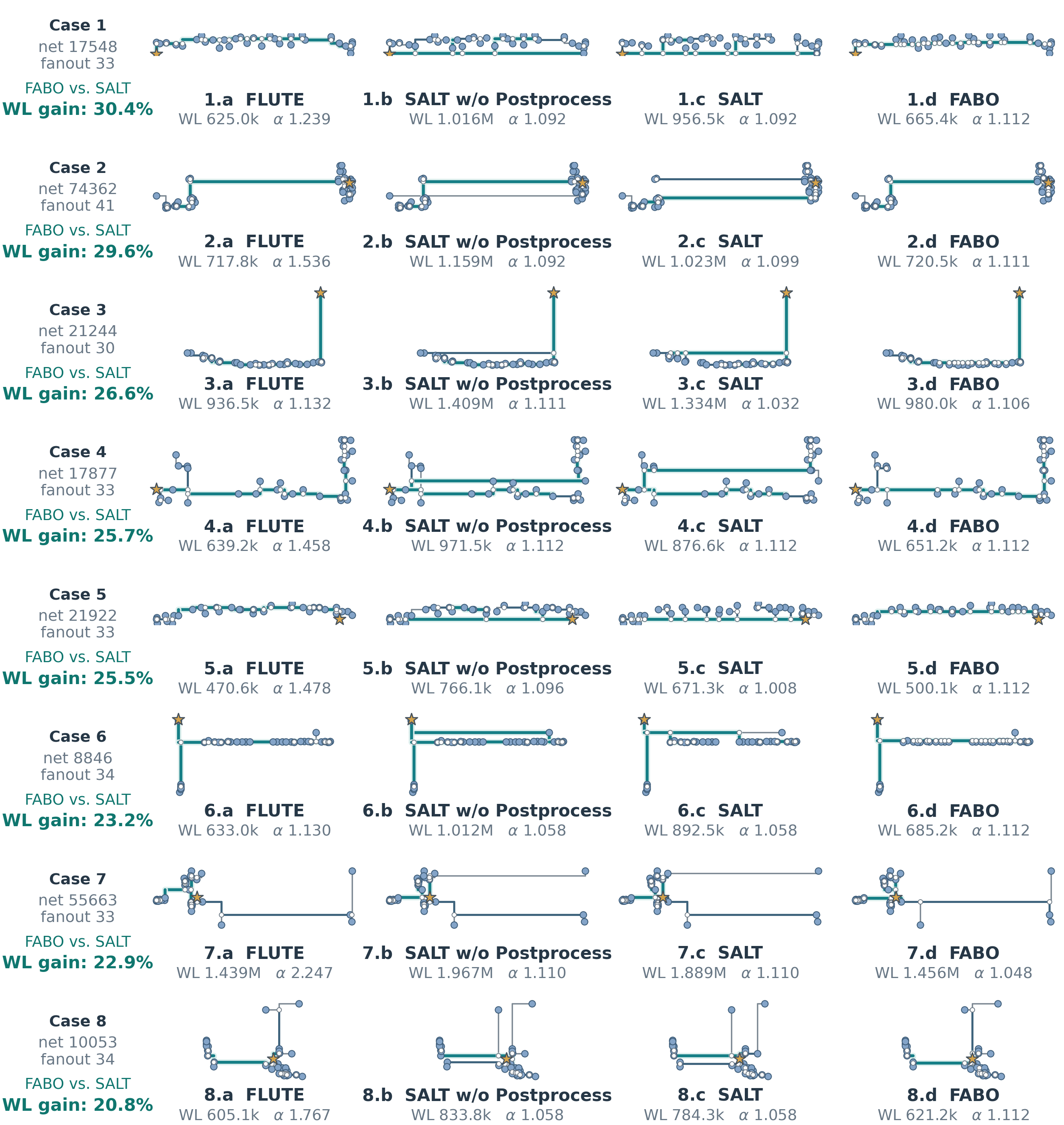}
\caption{\textbf{Eight high-degree examples of premature root-side splitting.}
All cases use $\epsilon=0.1125$.
The four columns show the FLUTE support, SALT before postprocessing, postprocessed SALT, and FABO.
The gold star marks the root pin $r$, filled blue circles mark sink pins, and hollow circles mark Steiner branch points.
Edge color indicates how many root-to-sink paths share each edge.
Panels in each row use the same plotting window and equal axis scale.
Each gain is FABO's wirelength reduction relative to final SALT on that row.}
\label{fig:fabo_high_gain_cases}
\end{figure*}

\subsection{Feasible Breakpoint Intervals}
\label{sec:fabo_ranges}

\textbf{One interval per flow.}
Let the light support tree $T^0$ (FLUTE~\cite{flute_iccad04} in our implementation) be rooted at the root pin $r$, and let $T^0[u,v]$ denote its unique path between points $u$ and $v$.
Placing the breakpoint of sink $t$ at a point $x\in T^0[r,t]$ means connecting $r$ to $x$ by a shortest rectilinear wire and keeping the support suffix from $x$ to $t$, as in Figure~\ref{fig:fabo_concepts}(b).
The new root-to-sink route then has length
\begin{equation}
s_t(x)=d_{T^0}(x,t),
\qquad
\phi_t(x)=d_M(r,x)+s_t(x).
\label{eq:fabo_phi}
\end{equation}
A point $x$ is \emph{feasible} for sink $t$ when $\phi_t(x)\le B_t$.
Here $\phi_t$ measures the route before any coincident wire merges; the realized tree path can only be shorter (Theorem~\ref{thm:fabo_feasible}), so feasibility under $\phi_t$ is a safe criterion.
Along $T^0[r,t]$, \emph{rootward} means closer to $r$ and \emph{deeper} means farther from $r$.

Moving $x$ rootward lengthens the kept suffix by exactly the distance moved, while the Manhattan connection $d_M(r,x)$ can shrink by at most that distance---and may even grow---so $\phi_t$ never decreases rootward.
The feasible points therefore form one suffix that starts at the rootward-most feasible point $f_t$:
\begin{equation}
I_t=\{x\in T^0[r,t]:\phi_t(x)\le B_t\}
    =T^0[f_t,t].
\label{eq:fabo_interval}
\end{equation}
We call $I_t$ the \emph{feasible breakpoint interval} of sink $t$ and of its flow.
If the whole support path is feasible, then $f_t=r$.
If $f_t$ falls inside an edge, FABO inserts it as a vertex, changing neither geometry nor wirelength.

\begin{lemma}[Feasible-interval property]
\label{lem:fabo_interval}
For every sink $t$, the feasible set in~\eqref{eq:fabo_interval} is a nonempty connected suffix of the support path $T^0[r,t]$.
\end{lemma}

\noindent\emph{Proof.}
Let $x'$ be rootward of $x$ on $T^0[r,t]$, and let $\Delta=d_{T^0}(x',x)$.
Since $x$ lies on $T^0[x',t]$, the suffix satisfies $s_t(x')=s_t(x)+\Delta$.
The triangle inequality gives $d_M(r,x')\ge d_M(r,x)-\Delta$.
Adding both, $\phi_t(x')\ge\phi_t(x)$: moving deeper never increases the path length.
The sink itself is feasible because $\phi_t(t)=d_M(r,t)\le B_t$, so the feasible set is exactly one nonempty suffix ending at $t$. \hfill$\square$

The lemma makes the interval exact for this criterion: $I_t$ contains precisely the breakpoints whose guaranteed route meets the budget.
The FABO lane of Figure~\ref{fig:algo_demo}, step~1, shows the resulting interval layout.
Every flow except $t_2$'s is feasible at every point of its support path ($f_t=r$; green sinks).
For $t_2$, the boundary $f_2$ satisfies $\phi_{t_2}(f_2)=623{,}384+44{,}004=667{,}388\le B_{t_2}$, leaving essentially no slack.

\subsection{Shared Breakpoint Selection}
\label{sec:fabo_selection}

\textbf{Select for the group, not for one sink.}
A breakpoint $a$ can serve sink $t$ exactly when $a\in I_t$.
If one point lies in several intervals, the corresponding flows can share one wire from the root to that point and separate only afterward, as in Figure~\ref{fig:fabo_concepts}(c).
FABO therefore selects a breakpoint set $A$ that intersects every interval,
\begin{equation}
A\cap I_t\ne\emptyset,
\qquad \forall t\in P,
\label{eq:fabo_cover}
\end{equation}
and assigns each sink the deepest selected point in its interval, keeping the suffix from that point:
\begin{equation}
a_t=\underset{a\in A\cap I_t}{\arg\max}\;\delta_0(a),
\qquad
S_t=T^0[a_t,t],
\label{eq:fabo_assignment}
\end{equation}
where $\delta_0(x)=d_{T^0}(r,x)$ is the support depth of point $x$.
The maximum in~\eqref{eq:fabo_assignment} is unique because all points of $A\cap I_t$ lie on the single path $T^0[r,t]$.

\textbf{Feasible initialization.}
FABO initializes $A$ greedily: it processes the \emph{proper} intervals---those with $f_t\ne r$---in nonincreasing depth $\delta_0(f_t)$ and adds the boundary $f_t$ whenever the current set misses $I_t$; it then adds the root, which covers every remaining interval.

\begin{theorem}[Minimum initial breakpoint set]
\label{thm:fabo_cover}
The boundaries added by the greedy pass form a minimum-size set among all sets that intersect every interval with $f_t\ne r$.
\end{theorem}

\noindent\emph{Proof.}
Call the interval that caused a point to be added the point's \emph{witness}; the greedy set has one point per witness.
It suffices to show that witnesses are pairwise disjoint, because any valid set must place a distinct point inside each interval of a disjoint family.
Suppose witnesses $I_i$ and $I_j$ share a point $x$, where $I_i$ was processed first, so $\delta_0(f_i)\ge\delta_0(f_j)$.
Because $x\in I_i$, the boundary $f_i$ is rootward of $x$, hence $f_i\in T^0[r,x]$; likewise $f_j\in T^0[r,x]$.
Points on one root path are totally ordered by depth, so $f_i$ lies on $T^0[f_j,x]$, which is contained in $I_j$.
Thus $f_i$ covered $I_j$ before $I_j$ was processed, contradicting that $I_j$ is a witness. \hfill$\square$

In the FABO lane of Figure~\ref{fig:algo_demo}, step~1 shows that only $t_2$'s interval excludes the root, so the greedy cover is $A=\{r,f_2\}$: sink $t_2$ is assigned $a_{t_2}=f_2$ and keeps its $44{,}004$ suffix, while every other flow is assigned the root and keeps its complete support path.
The cover thus involves no overlap decision; its contribution is placement.
$f_2$ keeps the longest suffix allowed by $t_2$'s budget. On this net, the same point also places the shared trunk as close as possible to the other six sinks without violating that budget.
Section~\ref{sec:fabo_full_tree} shows how the other flows consolidate onto this wire; overlap drives the selection when several flows are tight at once, where one boundary covers many intervals and Theorem~\ref{thm:fabo_cover} keeps the count minimum.

\textbf{Why selection does not stop here.}
The minimum cover is a guaranteed-feasible starting point, not the wirelength optimum: fewer breakpoints may shorten the new root connections, while deeper breakpoints may release more support wire.
FABO resolves this trade-off with the flow-aware account introduced next, and decides acceptance only on complete trees (Section~\ref{sec:fabo_full_tree}).

\subsection{Flow-Aware Refinement}
\label{sec:fabo_accounting}

\textbf{Count physical edges, not paths.}
FABO refines a breakpoint set $A$ by proposing one deeper feasible point $a$ at a time.
Under $A\cup\{a\}$, every sink is reassigned by the deepest-point rule of~\eqref{eq:fabo_assignment}:
\begin{equation}
a'_t(a)=
\underset{b\in(A\cup\{a\})\cap I_t}{\arg\max}\;\delta_0(b),
\qquad
S'_t(a)=T^0[a'_t(a),t].
\label{eq:fabo_proposed_suffix}
\end{equation}
The value of a proposal is the support wire it frees.
Summing every flow's shortened suffix would count one shared edge many times, so FABO accounts per edge.
For a retained support edge $e$, define its \emph{edge-user set}
\begin{equation}
Q_e=\{t\in P:e\in S_t\}
\label{eq:fabo_edge_users}
\end{equation}
and its \emph{leaving-flow set} under proposal $a$,
\begin{equation}
Y_e(a)=\{t\in P:e\in S_t\setminus S'_t(a)\}.
\label{eq:fabo_edge_movers}
\end{equation}
The edge is released exactly when every flow using it leaves (Figure~\ref{fig:fabo_concepts}(c)):
\begin{equation}
Y_e(a)=Q_e.
\label{eq:fabo_release_condition}
\end{equation}

\begin{lemma}[Exact release condition]
\label{lem:fabo_release}
After all flows are reassigned by~\eqref{eq:fabo_proposed_suffix}, a currently retained support edge $e$ belongs to no retained suffix if and only if~\eqref{eq:fabo_release_condition} holds.
\end{lemma}

\noindent\emph{Proof.}
Reassignment only moves breakpoints deeper, so $S'_t(a)\subseteq S_t$ and no flow gains edges; every flow using $e$ afterward already belongs to $Q_e$.
Edge $e$ survives exactly when some flow $t\in Q_e$ still uses it, that is, when $e\in S'_t(a)$ for some $t\in Q_e$---equivalently, when $Y_e(a)\ne Q_e$. \hfill$\square$

The released support length is therefore
\begin{equation}
G(a)=
\sum_{\substack{e:\,Q_e\ne\emptyset\\Y_e(a)=Q_e}}
\ell(e),
\label{eq:fabo_release}
\end{equation}
where $\ell(e)$ is the length of support edge $e$.

Figure~\ref{fig:algo_demo} shows why released wire must be counted per edge.
After SALT step~1, most of the old trunk remains because the other six flows still use it; only $126{,}540$ of the $804{,}065$ support wire becomes unused.
After FABO step~2 moves all seven flows to the new spine, the old trunk has no users and can be removed entirely.
In this example, standard postprocessing performs the migration; the edge-level gain defined above is used to rank candidate breakpoints during refinement.
The gain $G(a)$ only ranks proposals; acceptance is decided on the complete tree, where the new root connections are also paid.

\subsection{Complete-Tree Construction and Acceptance}
\label{sec:fabo_full_tree}

\textbf{Evaluate the object that will be returned.}
For a breakpoint set $A$, FABO builds a rectilinear \emph{connector} $C(A)$ that spans the root and the selected breakpoints and reaches every breakpoint at its Manhattan distance from the root:
\begin{equation}
d_{C(A)}(r,a)=d_M(r,a),
\qquad \forall a\in A.
\label{eq:fabo_connector}
\end{equation}
Grafting every retained suffix $S_t$ under its breakpoint $a_t$ and removing redundant Steiner nodes yields the raw tree $T_{\mathrm{raw}}(A)$.

\begin{theorem}[Feasibility of a realized candidate]
\label{thm:fabo_feasible}
If every sink $t$ is assigned a breakpoint $a_t\in I_t$, the connector satisfies~\eqref{eq:fabo_connector}, and realization preserves every retained suffix and returns a tree, then $T_{\mathrm{raw}}(A)$ satisfies every budget in~\eqref{eq:salt_formulation}.
\end{theorem}

\noindent\emph{Proof.}
For each sink $t$, the realized tree contains the connector path from $r$ to $a_t$ followed by the suffix $S_t$, and the tree path from $r$ to $t$ is no longer than this walk:
\begin{equation}
d_{T_{\mathrm{raw}}(A)}(r,t)
\le d_M(r,a_t)+s_t(a_t)
=\phi_t(a_t)\le B_t,
\label{eq:fabo_path_bound}
\end{equation}
where the last step uses $a_t\in I_t$. \hfill$\square$

\textbf{Realization and postprocessing.}
$\textsc{FullRealize}(A)$, used on finalist candidates, builds the connector with SALT's rectilinear Steiner arborescence construction~\cite{salt_iccad17}; $\textsc{FastRealize}(A)$, used while scanning proposals, builds it by merging overlapping shortest root paths.
Both constructions satisfy~\eqref{eq:fabo_connector}; suffix preservation and tree validity are enforced by construction and re-checked by the evaluator, and merging coincident wire only shortens root-to-sink paths, so the bound of Theorem~\ref{thm:fabo_feasible} survives.
Each realized tree is additionally passed through SALT's fixed postprocessing, which may re-attach a flow onto the connector when doing so shortens the tree without violating any budget; the raw and postprocessed forms are both evaluated, and $\textsc{FullRealize}$ and $\textsc{FastRealize}$ return the better valid form.
Each of these local moves relocates one attachment, so the passes do not assemble a long new trunk on their own---a coordinated migration would have to cross non-improving intermediate trees---but they readily exploit a trunk that exists: moving a budget-loose drop onto it is a single improving legal move.
The same postprocessing therefore stalls on SALT's tree and completes the consolidation on FABO's.

\textbf{Acceptance.}
The evaluator requires a connected, acyclic rooted tree that contains every pin, checks every root-to-sink budget, and measures exact wirelength.
Among valid trees, lower wirelength wins, and maximum stretch $\alpha$ breaks a wirelength tie; we write this order as $\prec_{\mathrm{lex}}$. The FABO lane of Figure~\ref{fig:algo_demo} separates these roles: step~1 selects $f_2$ and builds its shortest root connector, while step~2 moves the other six flows onto the resulting spine, leaving them on shortest root paths and deleting the old trunk. FABO step~3 of Figure~\ref{fig:algo_demo} shows the accepted result: every budget holds, the path of $t_2$ is $667{,}388$ against its budget of $667{,}388.75$, and wirelength is $929{,}794$ against SALT's $1{,}042{,}970$ under the same constraints.

\subsection{FABO Candidate Search}
\label{sec:fabo_algorithm}

\textbf{Two-stage search.}
Algorithm~\ref{alg:fabo} summarizes FABO on one support tree.
Stage one computes the intervals, builds the minimum cover, and realizes it; Theorem~\ref{thm:fabo_feasible} guarantees its budgets and the evaluator confirms them, so this stage already yields a valid tree.
Stage two proposes a bounded number of deeper breakpoints, ranks them by a score dominated by the flow-aware gain, commits a proposal only when its fast-realized tree is valid and shorter, and finally re-realizes the enlarged set in full.

\begin{algorithm}[!t]
\caption{\textbf{FABO on one light support tree.}}
\label{alg:fabo}
\begin{algorithmic}[1]
\REQUIRE Root $r$, sink set $P$, path budgets $B_t$, support tree $T^0$
\ENSURE A feasible rectilinear routing tree
\STATE Compute $f_t$ and $I_t=T^0[f_t,t]$ for every sink $t\in P$
\STATE $A\gets\emptyset$
\FOR{each interval with $f_t\ne r$, in nonincreasing $\delta_0(f_t)$}
  \IF{$A\cap I_t=\emptyset$}
    \STATE $A\gets A\cup\{f_t\}$
  \ENDIF
\ENDFOR
\STATE $A\gets A\cup\{r\}$
\STATE $T_\star\gets\textsc{FullRealize}(A)$;\;
       $T_{\mathrm{scan}}\gets\textsc{FastRealize}(A)$
\STATE Form a bounded pool $Q$ of deeper points from the feasible intervals
\FOR{a bounded number of refinement rounds}
  \STATE Rank every $a\in Q\setminus A$ by a fixed score dominated by $G(a)$
  \STATE Fast-realize and validate the leading proposals $A\cup\{a\}$
  \IF{no valid proposal improves $T_{\mathrm{scan}}$}
    \STATE \textbf{break}
  \ENDIF
  \STATE Commit the best proposal $a^\star$: $A\gets A\cup\{a^\star\}$; update $T_{\mathrm{scan}}$
\ENDFOR
\STATE $T'\gets\textsc{FullRealize}(A)$
\IF{$T'\prec_{\mathrm{lex}}T_\star$}
  \STATE $T_\star\gets T'$
\ENDIF
\RETURN $T_\star$
\end{algorithmic}
\end{algorithm}

\begingroup
\def\FABOResultFiguresOnly{1}
\ifdefined\FABOResultFiguresOnly
\begin{figure*}[!t]
\centering
\begin{subfigure}[t]{0.245\textwidth}
  \centering
  \includegraphics[width=\textwidth]{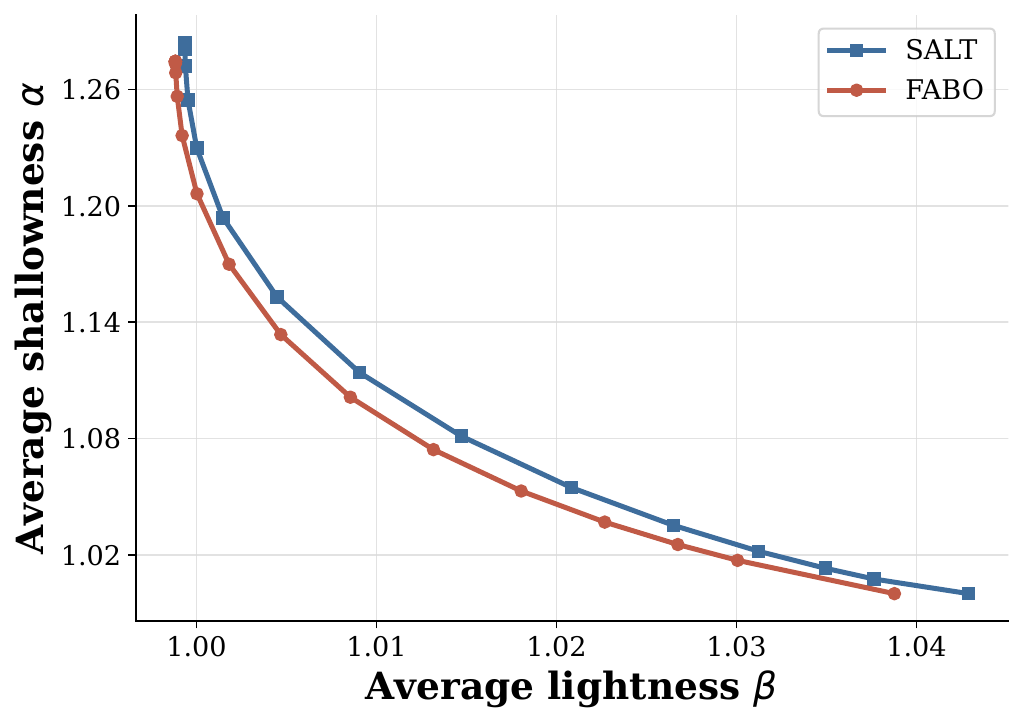}
  \caption{$\alpha$--$\beta$}
  \label{fig:fabo_salt_alpha_beta}
\end{subfigure}\hfill
\begin{subfigure}[t]{0.245\textwidth}
  \centering
  \includegraphics[width=\textwidth]{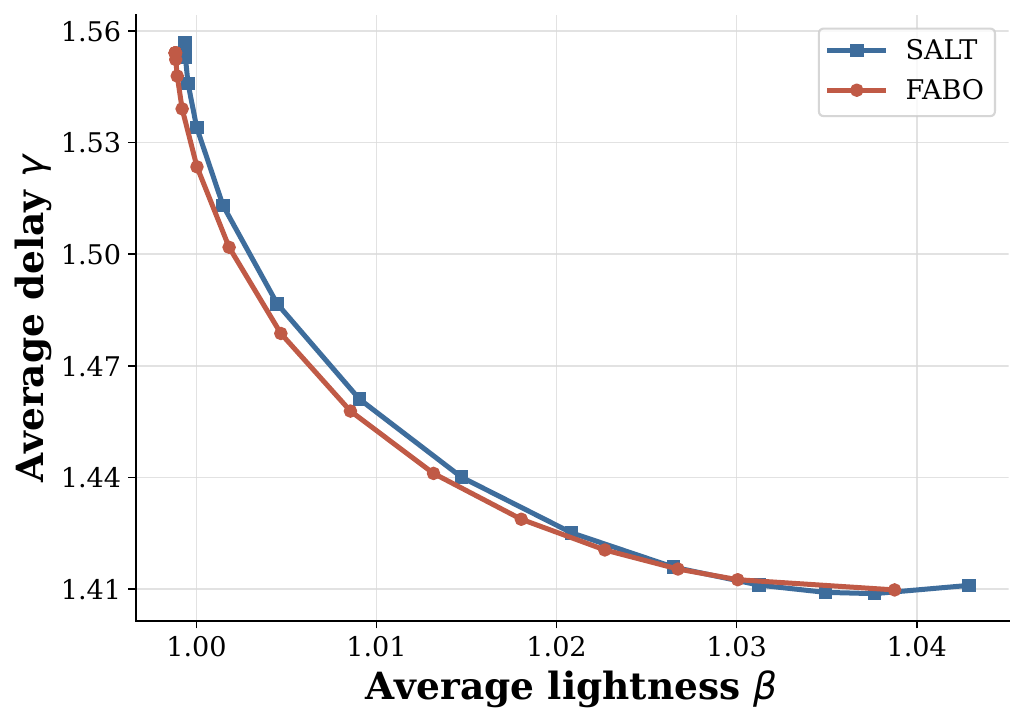}
  \caption{$\gamma$--$\beta$}
  \label{fig:fabo_salt_gamma_beta}
\end{subfigure}\hfill
\begin{subfigure}[t]{0.245\textwidth}
  \centering
  \includegraphics[width=\textwidth]{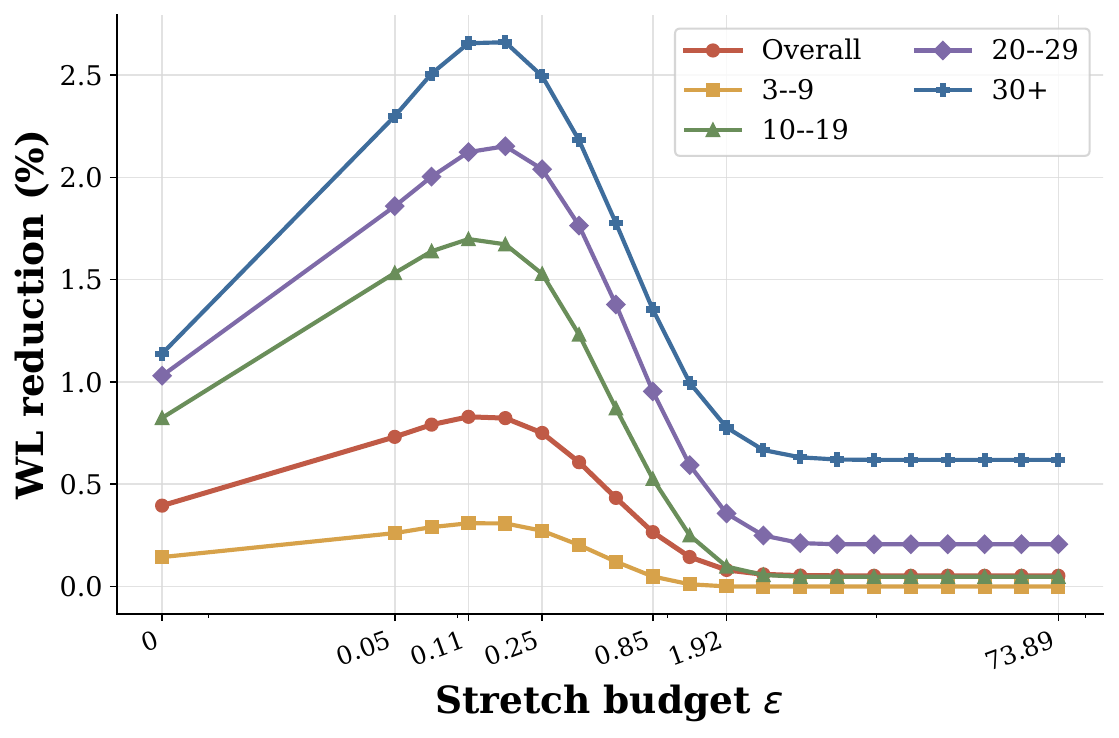}
  \caption{WL reduction vs.\ $\epsilon$}
  \label{fig:fabo_salt_wl_reduction_eps}
\end{subfigure}\hfill
\begin{subfigure}[t]{0.245\textwidth}
  \centering
  \includegraphics[width=\textwidth]{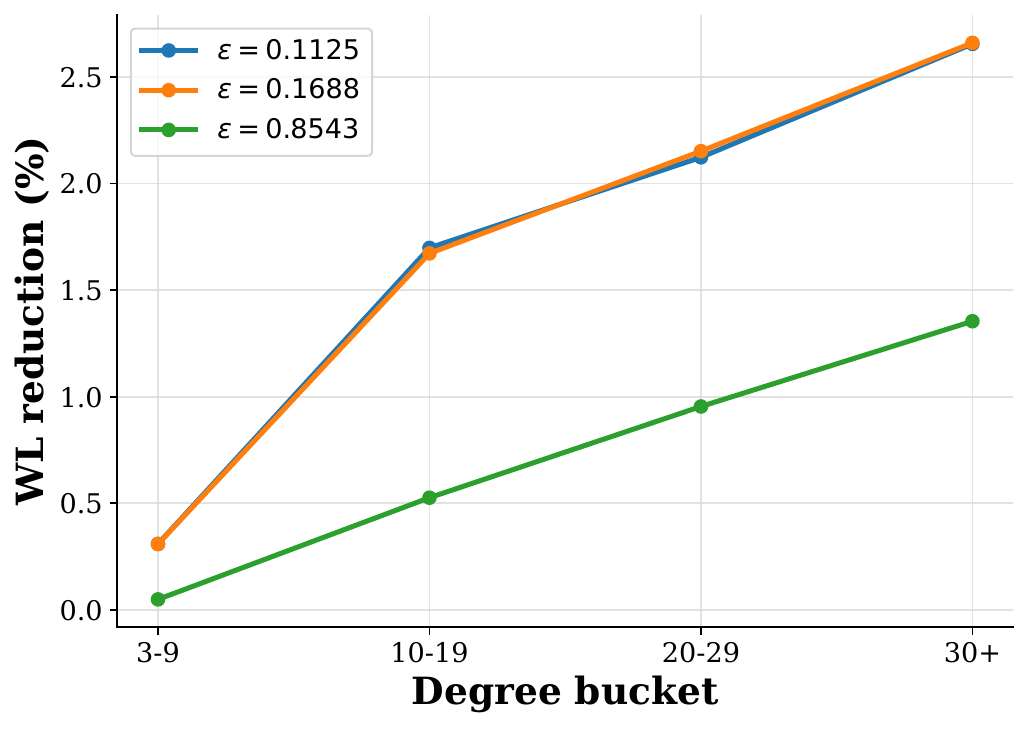}
  \caption{WL reduction vs.\ degree}
  \label{fig:fabo_salt_wl_reduction_degree}
\end{subfigure}
\caption{\textbf{FABO dominates SALT across all 20 stretch tolerances.}
Panels (a) and (b) show full-benchmark means at each $\epsilon$; FABO lies below and to the left of SALT in both plots.
Panel (c) reports FABO's full-benchmark wirelength reduction at each $\epsilon$.
Panel (d) reports the reduction within each pin-count range at the three shown tolerances.}
\label{fig:fabo_salt_main_results}
\end{figure*}

\begin{figure*}[!t]
\centering
\begin{subfigure}[t]{0.245\textwidth}
  \centering
  \includegraphics[width=\textwidth]{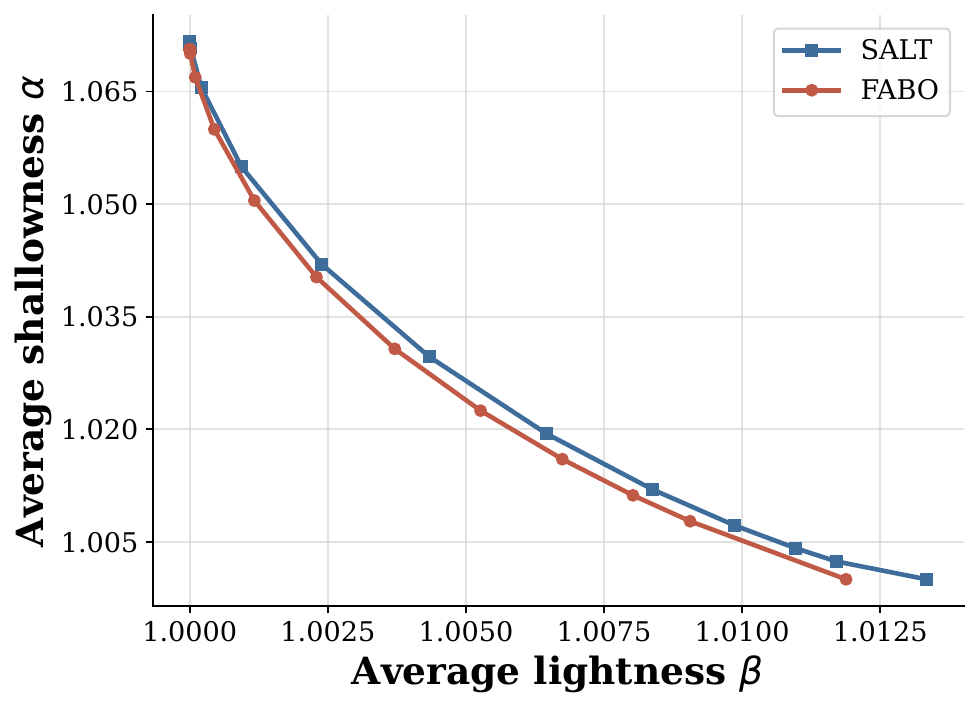}
  \caption{3--9 pins}
\end{subfigure}\hfill
\begin{subfigure}[t]{0.245\textwidth}
  \centering
  \includegraphics[width=\textwidth]{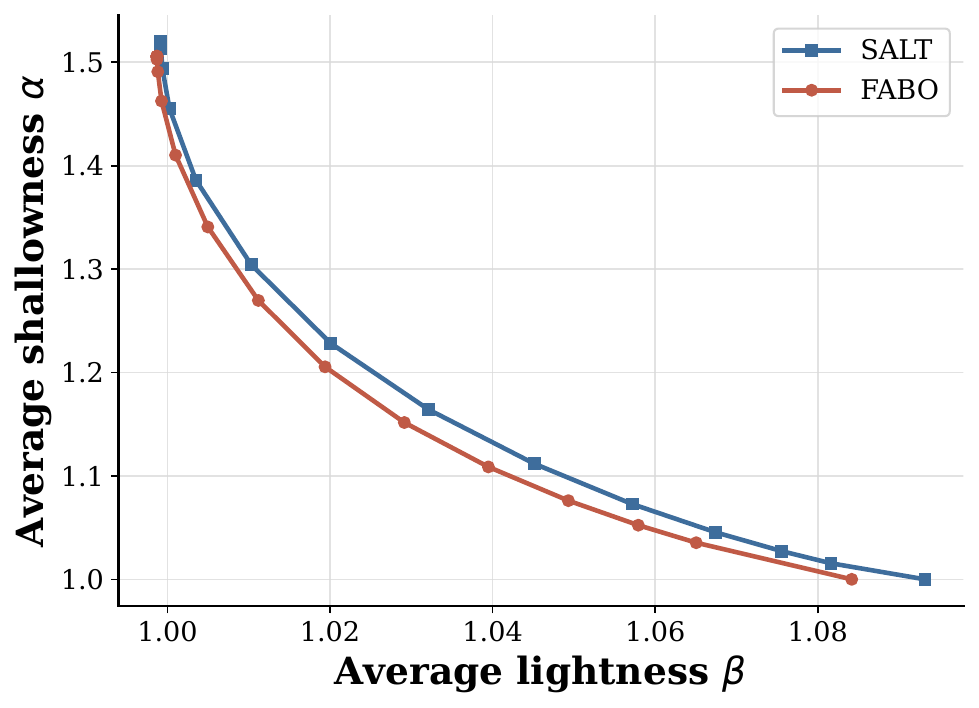}
  \caption{10--19 pins}
\end{subfigure}\hfill
\begin{subfigure}[t]{0.245\textwidth}
  \centering
  \includegraphics[width=\textwidth]{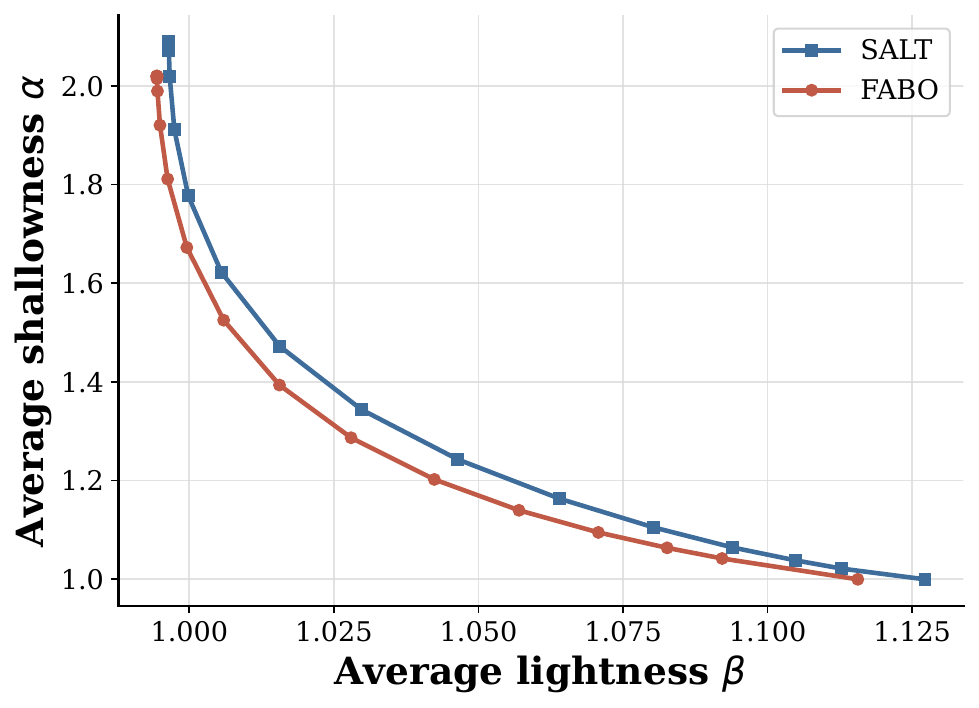}
  \caption{20--29 pins}
\end{subfigure}\hfill
\begin{subfigure}[t]{0.245\textwidth}
  \centering
  \includegraphics[width=\textwidth]{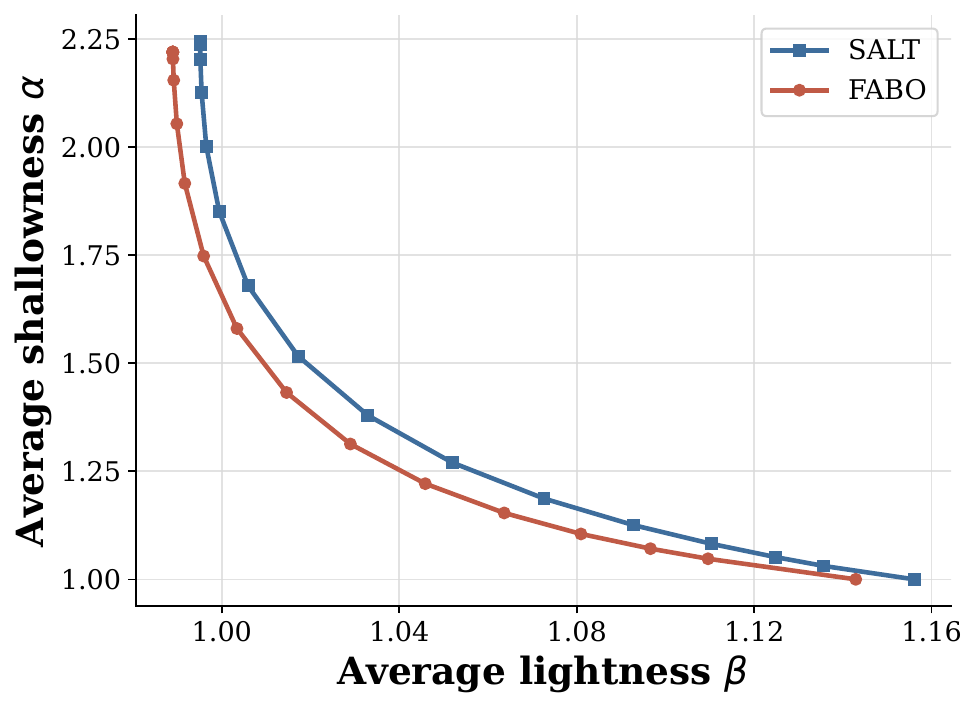}
  \caption{30+ pins}
\end{subfigure}
\caption{\textbf{FABO dominates SALT in every pin-count range, with a wider gap at higher fanout.}
Each point averages $\alpha$ and $\beta$ over all nets in the labeled pin-count range for one method and one $\epsilon$.}
\label{fig:fabo_salt_degree_buckets}
\end{figure*}

\begin{figure*}[!t]
\centering
\begin{subfigure}[t]{0.245\textwidth}
  \centering
  \includegraphics[width=\textwidth]{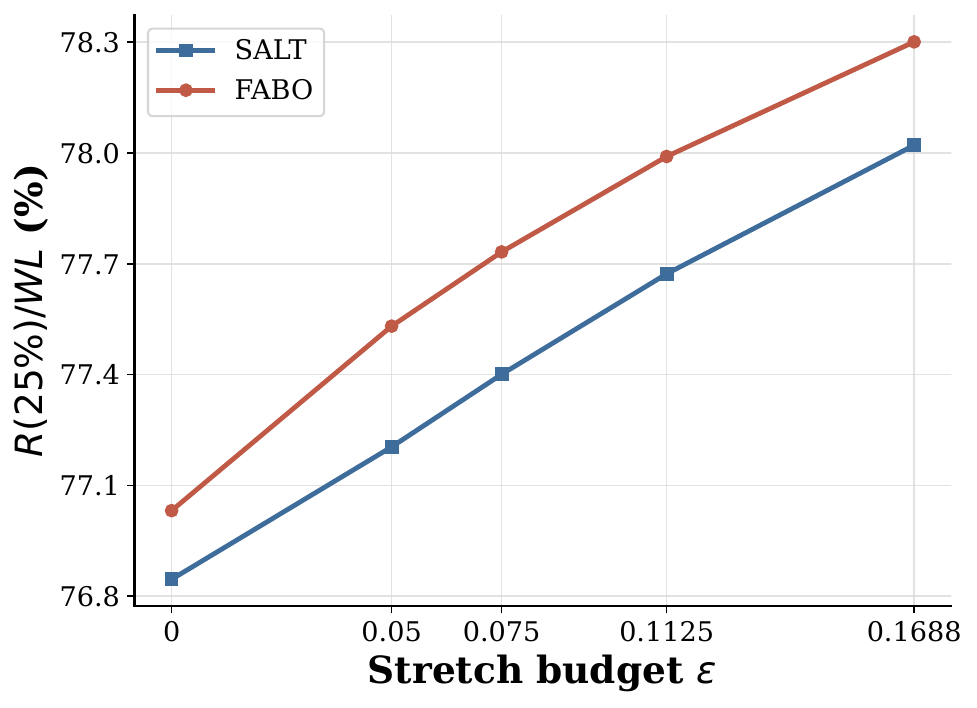}
  \caption{$R_T(25\%)/WL(T)$}
\end{subfigure}\hfill
\begin{subfigure}[t]{0.245\textwidth}
  \centering
  \includegraphics[width=\textwidth]{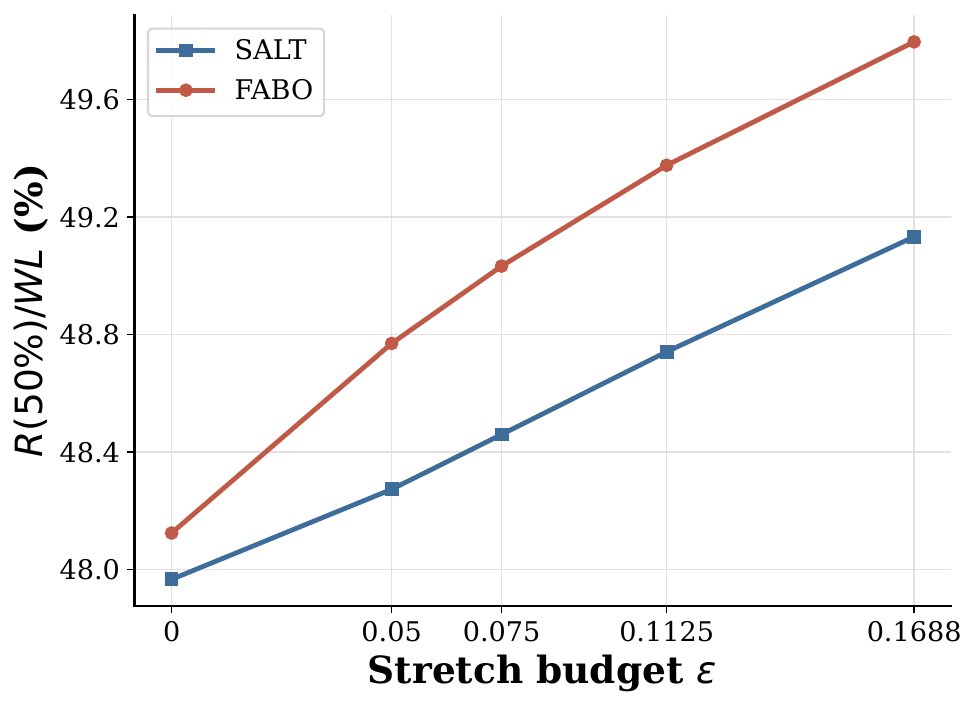}
  \caption{$R_T(50\%)/WL(T)$}
\end{subfigure}\hfill
\begin{subfigure}[t]{0.245\textwidth}
  \centering
  \includegraphics[width=\textwidth]{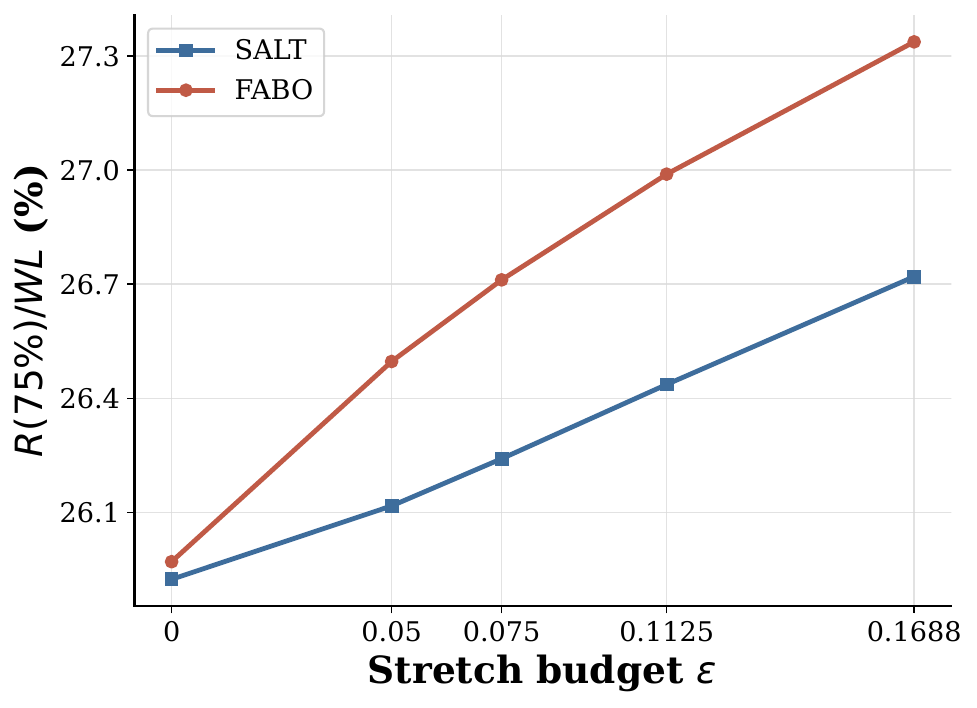}
  \caption{$R_T(75\%)/WL(T)$}
\end{subfigure}\hfill
\begin{subfigure}[t]{0.245\textwidth}
  \centering
  \includegraphics[width=\textwidth]{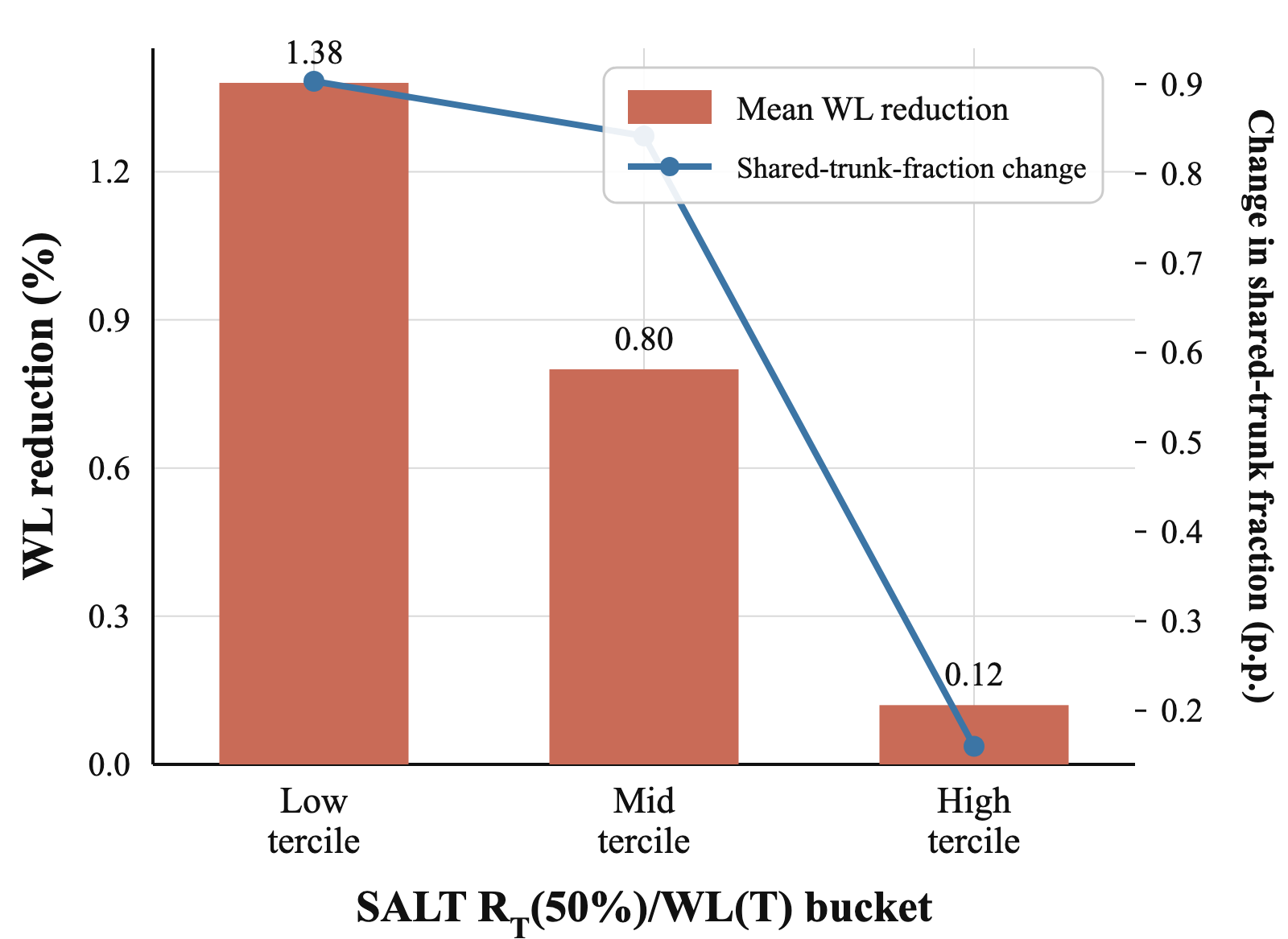}
  \caption{Grouped by SALT sharing}
\end{subfigure}
\caption{\textbf{FABO keeps a larger fraction of wirelength on shared root-side trunks.}
The metric $R_T(\rho)/WL(T)$ is the fraction of total wirelength in the root-connected trunk used by at least $\rho$ of the net's root-to-sink paths.
Panels (a)--(c) report the mean sharing metric over all 1.29 million nets at the five shown tolerances.
Panel (d) groups nets into terciles by SALT's $R_T(50\%)/WL(T)$ at $\epsilon=0.1125$. The bars show FABO's mean wirelength reduction relative to SALT; the line shows FABO's mean shared-trunk fraction minus SALT's.}
\label{fig:fabo_salt_structure_validation}
\end{figure*}

\begin{figure*}[!t]
\centering
\includegraphics[width=0.84\textwidth]{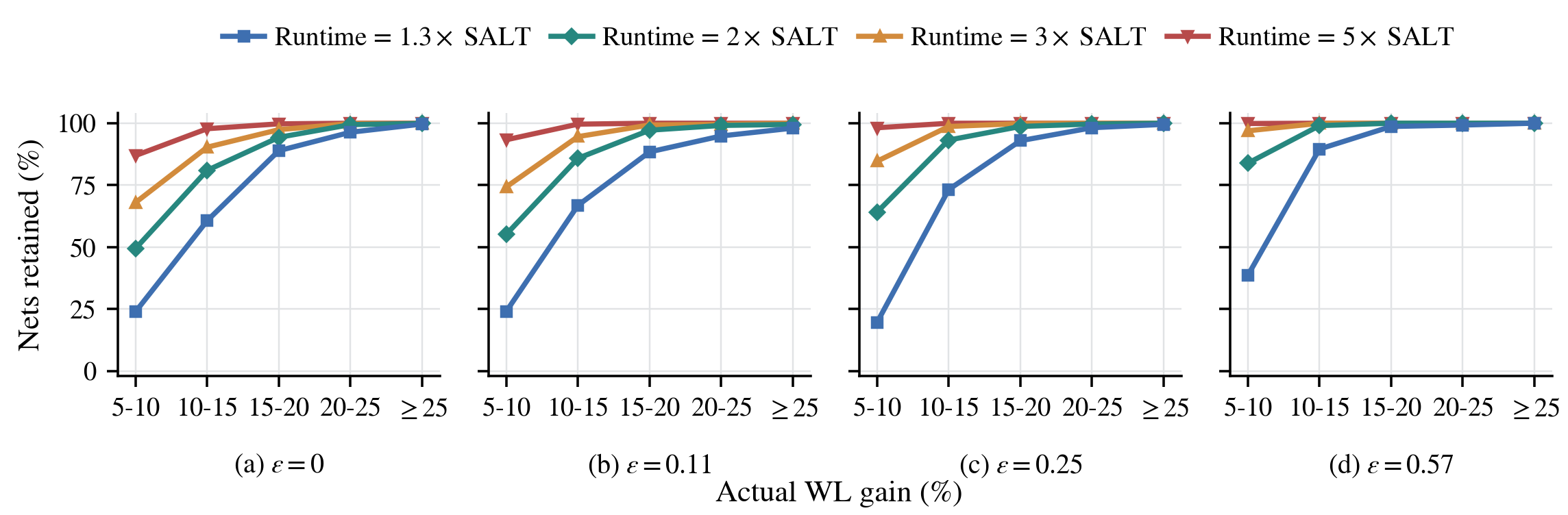}
\caption{\textbf{FABO-FAST recovers the high-gain nets first.}
The x-axis bins nets by the wirelength reduction obtained when FABO runs on every net. The y-axis is the fraction of each bin selected by FABO-FAST. Each curve represents a total runtime budget normalized by the all-net SALT runtime; the budget includes prediction and FABO calls on selected nets.
The four panels use $\epsilon=0$, 0.1125, 0.253125, and 0.56953125, with a separate leave-one-design-out predictor trained for each $\epsilon$.}
\label{fig:selective_fabo_budget_quality}
\end{figure*}
\else
\section{Results}
\label{sec:results}

\subsection{Experimental Setup}
\label{sec:result_setup}

\textbf{Benchmark and baseline.}
We evaluate FABO on the eight ICCAD15 Contest Problem B designs distributed by the official SALT repository flow~\cite{iccad15_contest,salt_iccad17}: \texttt{superblue1}, \texttt{superblue3}, \texttt{superblue4}, \texttt{superblue5}, \texttt{superblue7}, \texttt{superblue10}, \texttt{superblue16}, and \texttt{superblue18}.\footnote{The released benchmark differs from the counts reported in the SALT paper. It contains no 2-pin nets and 1.29M nets with at least three pins, whereas~\cite{salt_iccad17} reports 2.38M nontrivial nets. Most of the difference is in the 3--9-pin bucket. This discrepancy does not affect our comparison: FABO and the official SALT implementation use the same released files and the same evaluator.}
The benchmark contains 1,294,541 nets with at least three pins.
We use postprocessed SALT as the direct baseline because it targets the same objective: minimum wirelength under the same per-sink path-length limit. For every net and tolerance, FABO and SALT receive the same input and are evaluated by the same checker.
Because 3--9-pin nets make up 69.9\% of the benchmark, while only 58.3k nets have at least 30 pins, we report both full-benchmark and fanout-stratified results (Table~\ref{tab:iccad15_benchmark_stats}).

\begin{table}[t]
\caption{\textbf{Benchmark size and fanout distribution.}
Net degree is the number of pins; each entry counts nets directly from the eight ICCAD15 \texttt{.nets} files.}
\label{tab:iccad15_benchmark_stats}
\centering
\small
\resizebox{\columnwidth}{!}{%
\begin{tabular}{lrrrrrr}
\toprule
Design & 3--9 & 10--19 & 20--29 & 30--39 & $\geq 40$ & $\geq 3$ \\
& ($\times 10^3$) & ($\times 10^3$) & ($\times 10^3$) & ($\times 10^3$) & ($\times 10^3$) & ($\times 10^3$) \\
\midrule
superblue1 & 135 & 23.3 & 11.2 & 5.8 & 0.9 & 176.2 \\
superblue10 & 147.7 & 31.3 & 13.8 & 9.5 & 1.2 & 203.5 \\
superblue16 & 114 & 17.2 & 7.3 & 5.3 & 0.3 & 144.2 \\
superblue18 & 54.5 & 24.5 & 10.9 & 5.1 & 0.6 & 95.5 \\
superblue3 & 102.1 & 35 & 15.4 & 6.1 & 1.1 & 159.8 \\
superblue4 & 74.2 & 16.9 & 8.7 & 3.8 & 0.5 & 104 \\
superblue5 & 123.2 & 18.2 & 7.6 & 4.7 & 0.7 & 154.3 \\
superblue7 & 154.4 & 62.5 & 27.5 & 11 & 1.7 & 257.1 \\
\midrule
Total & 904.9 & 228.8 & 102.5 & 51.3 & 7 & 1294.5 \\
\bottomrule
\end{tabular}
}

\end{table}

\textbf{Metrics.}
The stretch tolerance $\epsilon$ limits every root-to-sink path to $(1+\epsilon)$ times its Manhattan distance.
We use three quality metrics. For each sink, path stretch is its root-to-sink path length divided by its Manhattan distance; shallowness $\alpha$ is the maximum of these ratios. FLUTE-normalized wirelength $\beta=WL(T)/WL(T_{\mathrm{FLUTE}})$ is total tree wirelength divided by the FLUTE wirelength for the same net. Delay $\gamma$ is the normalized maximum Elmore-delay metric used by SALT.
Lower is better for all three.
At fixed $\epsilon$, both algorithms satisfy the same per-sink limit, so lower mean wirelength means lower routing cost under the same feasibility constraint.

\subsection{FABO Dominates SALT at Every Constraint}
\label{sec:result_overall}

\textbf{Main result.}
At each of the 20 stretch-tolerance settings, FABO has lower mean wirelength than postprocessed SALT.
The largest full-benchmark reduction is 0.83\% at $\epsilon=0.1125$: mean FLUTE-normalized wirelength is 1.0312 for SALT and 1.0227 for FABO.

\textbf{Shape of the sweep.}
At the strictest constraint, $\epsilon=0$, both methods reach $\alpha=1$. SALT's $\beta$ is 1.0429, whereas FABO's is 1.0388.
The reduction peaks at 0.83\% once the constraint allows enough slack to move shared breakpoints jointly across several paths. At larger $\epsilon$, both methods approach the low-wirelength FLUTE topology, so FABO's advantage narrows but remains positive.

\textbf{Pareto dominance.}
At every $\epsilon$, FABO has lower mean $\beta$ without increasing mean $\alpha$ or $\gamma$; FABO therefore dominates SALT at all 20 operating points in both trade-off plots.
In the sampled sweep, FABO first reaches $\beta<1.001$ at $\alpha=1.2062$, whereas SALT first reaches it at $\alpha=1.2298$.
On the delay frontier, FABO first reaches $\beta<1.001$ at $\gamma=1.5235$, whereas SALT first reaches it at $\gamma=1.5340$.

\subsection{The Gain Grows with Fanout}
\label{sec:result_direct}
\label{sec:result_degree}

\textbf{Fanout trend.}
FABO dominates SALT in all four pin-count ranges, and its advantage grows with net degree.
For each $\epsilon$, we report FABO's percentage reduction in mean FLUTE-normalized wirelength relative to SALT.
Across the 20 tolerances, the maximum reduction is 0.31\% for 3--9 pins, 1.70\% for 10--19, 2.15\% for 20--29, and 2.66\% for at least 30 pins. At each of the three tolerances shown in Figure~\ref{fig:fabo_salt_main_results}(d), the reduction increases across the four pin-count ranges.
Because 3--9-pin nets make up 69.9\% of the benchmark, they dominate the full-benchmark average; this is why the overall peak of 0.83\% is much smaller than the high-fanout peaks.
Higher-fanout nets have more root-side edges used by several sink paths. SALT repairs each sink path separately, which can duplicate this shared wire; FABO instead moves shared breakpoints jointly across paths. The next subsection tests whether preserving shared root-side wire explains FABO's larger gains on high-fanout nets.

\subsection{Why FABO Saves Wire: A Larger Shared-Trunk Fraction}
\label{sec:result_mechanism}

\textbf{Mechanism test.}
If FABO's gain comes from preserving shared root-side wire, its trees should keep a larger fraction of their wirelength on heavily shared trunks.
For a threshold $\rho$, consider the root-connected trunk formed by edges used by at least $\rho$ of the net's root-to-sink paths. We denote its length by $R_T(\rho)$ and use $R_T(\rho)/WL(T)$ as the shared-trunk fraction.
We compute the sharing metric over all 1.29 million nets at the five shown tolerances.

For every shown pair of $\epsilon$ and $\rho$, FABO's mean shared-trunk fraction is higher than SALT's.
At $\epsilon=0.1125$ and $\rho=50\%$, SALT's mean shared-trunk fraction is 48.74\%, whereas FABO's is 49.38\%.
Panel (d) divides nets into terciles using SALT's $R_T(50\%)/WL(T)$. Relative to SALT, FABO reduces mean wirelength by 1.38\% in the low-sharing tercile, 0.80\% in the middle tercile, and 0.12\% in the high-sharing tercile. The largest reduction therefore occurs on nets for which SALT preserves the least shared wire.
Together, the fanout and trunk results support the proposed mechanism: FABO's trees allocate a larger fraction of their wirelength to shared root-side trunks, especially on the nets where SALT shares the least.

\subsection{FABO-FAST: Budget-Controlled FABO}
\label{sec:result_selective}

\textbf{Why FABO-FAST.}
Running FABO on every net gives the lowest mean wirelength, but at $\epsilon=0.1125$ the complete flow takes $10.66\times$ the all-net SALT runtime.
FABO-FAST first runs SALT on every net. A learned predictor ranks nets by their expected FABO gain, and FABO runs in that order until the total runtime reaches the budget. An unselected net retains its SALT tree. For a selected net, the final output is the shorter feasible tree produced by SALT and FABO.

\textbf{Gain and recovery.}
For net $i$, let $S_i$ and $F_i$ be the wirelengths of its feasible SALT and FABO trees. Because the final output uses the shorter tree, we define the realized percentage reduction as
\begin{equation}
g_i=100\frac{\max(S_i-F_i,0)}{S_i}.
\label{eq:selective_actual_gain}
\end{equation}
For a gain interval $I$, such as $[20\%,25\%)$, let $\mathcal{A}$ be the nets selected by FABO-FAST. We define recovery as the fraction of all nets whose full-FABO gain lies in $I$ that FABO-FAST selects:
\begin{equation}
Q_I(\mathcal{A})=100\frac{|\{i:i\in\mathcal{A},\ g_i\in I\}|}{|\{i:g_i\in I\}|}.
\label{eq:selective_interval_retention}
\end{equation}
The denominator is computed by running FABO on every net at the same $\epsilon$. Thus, $Q_I=100\%$ means FABO-FAST selects every net whose full-FABO gain falls in interval $I$.

\textbf{Predictor.}
The depth-4, 96-tree XGBoost predictor takes 22 features computed before FABO runs.
These features describe pin count, bounding-box geometry, FLUTE wirelength and path statistics, and SALT wirelength, stretch, and delay.
We use eight-fold leave-one-design-out evaluation. Each fold trains on seven designs and predicts all nets in the held-out design, so every net receives exactly one out-of-design prediction. We train a separate predictor for each plotted $\epsilon$.

\textbf{Runtime accounting.}
Let $T_{\mathrm{SALT}}$ be the time to run SALT on every net. The FABO-FAST runtime $T_{\mathrm{FAST}}$ includes $T_{\mathrm{SALT}}$, one batch prediction, and FABO calls on selected nets. We report the budget multiplier $b=T_{\mathrm{FAST}}/T_{\mathrm{SALT}}$.
At $\epsilon=0.1125$, SALT takes 62.43~s, running FABO on all nets adds 602.83~s, and batch prediction takes 0.872~s.

\textbf{Runtime--quality trade-off.}
At $b=1.3$, across the four plotted $\epsilon$ values, FABO-FAST recovers 99.0\%--100\% of nets with at least 25\% gain and 96.0\%--99.3\% of nets with 20\%--25\% gain. Recovery is lower for smaller gains: 88.5\%--99.2\% for the 15\%--20\% group and 21.0\%--45.3\% for the 5\%--10\% group. The monotonic drop in recovery as gain decreases shows that the predictor ranks the highest-gain nets first.
At $b=2$, recovery is 98.9\%--100\% for both the 20\%--25\% group and the at-least-25\% group. As $b$ increases, FABO-FAST selects more nets and approaches the result of running FABO on every net.

\textbf{Why high-gain nets are predictable.}
A large increase in wirelength from FLUTE to SALT, or a large change in path stretch, suggests that SALT made an expensive structural repair to satisfy the path constraint. Such nets give FABO's coordinated breakpoints the most opportunity to reduce wirelength.

% In a held-out ablation that adds feature groups cumulatively, the mean actual gain among the top 1\% of predicted nets is 2.76\% with pin count alone, 3.12\% after adding bounding-box geometry, 4.05\% after adding FLUTE path summaries, and 9.93\% after adding SALT summaries.
% The largest improvement comes from adding SALT-tree features.

% \textbf{Absolute scale.}
% At $\epsilon=0.1125$, full FABO finds 649 nets with at least 20\% gain, 202 with at least 25\%, and 52 with at least 30\%.
% At $1.04\times$ SALT time, FABO-FAST selects 511 of the 649 nets with at least 20\% gain and 50 of the 52 nets with at least 30\% gain. At $1.50\times$ SALT time, FABO-FAST selects 636 of 649 nets with at least 20\% gain, 201 of 202 with at least 25\% gain, and all 52 with at least 30\% gain.

% \textbf{Timing robustness.}
% For the selection counts in the previous paragraph, we estimate each net's runtime by averaging two to four repeated runs with the same configuration. Figure~\ref{fig:selective_fabo_budget_quality} instead uses the raw per-net runtime recorded in each sweep.
% To test whether rare slow runs drive the result, we cap per-net FABO runtime at the 99.9th percentile and recompute recovery. Across the four plotted $\epsilon$ values, recovery changes by at most 4.10 percentage points; across the three nonzero values, it changes by at most 0.49 percentage points. The actual runtime for every reported budget remains within $0.003\times T_{\mathrm{SALT}}$ of its target. Thus, the concentration of high-gain nets is not caused by a few runtime outliers.
\fi

\endgroup

\textbf{Search scope.}
Every bound in the search is a fixed constant, identical for all nets and tolerances and released with the code: the pool keeps a constant number of candidates, each round fast-realizes only the top-ranked proposals, and a round with little gain ends the search early.
The ranking is a fixed weighted score dominated by the released length $G(a)$, plus two lighter terms: one rewards the total depth by which reassigned flows move, and one discounts proposals far from every already selected breakpoint---a proxy for the connector wire the addition will need.
On nets with at least 31 pins, FABO repeats the whole search on a second, higher-accuracy FLUTE support tree and returns the better valid result.

\subsection{Guarantees and Scope}
\label{sec:fabo_guarantees}

\textbf{What is guaranteed.}
Lemma~\ref{lem:fabo_interval} gives exact feasible intervals; Theorem~\ref{thm:fabo_cover} minimizes only the number of initial non-root breakpoints; Theorem~\ref{thm:fabo_feasible} proves budget feasibility if the connector and realization satisfy their stated conditions; and Lemma~\ref{lem:fabo_release} makes edge-release accounting exact. These are properties of the construction, not the final acceptance guarantee.
In the implementation, a fixed evaluator rejects any incomplete or budget-violating candidate. The validated stage-one tree initializes the incumbent, which is replaced only by a valid candidate that improves $(WL,\alpha)$ under $\prec_{\mathrm{lex}}$.
Thus, the evaluator and incumbent-update rule ensure that FABO returns a feasible tree no worse than its validated stage-one tree.

\textbf{What remains heuristic.}
The minimum cover minimizes only the number of initial non-root breakpoints; it does not guarantee minimum final wirelength. Likewise, the gain $G(a)$ measures released support wire, not the net wirelength change after adding the connector.
Candidate ranking, connector construction, support-tree choice, and postprocessing are therefore heuristics; complete-tree evaluation is the only acceptance rule.

\section{Conclusion}

We developed an agent-guided framework that combines parallel language-model exploration with independent checking to discover and validate routing-tree algorithms.
Applied to SALT, it discovered that sink-by-sink repair cannot coordinate breakpoints across paths sharing root-side wire, causing premature splitting and duplicated wire.
This finding led to FABO, which represents each path's feasible breakpoint choices as an interval and jointly selects shared breakpoints while preserving every sink's stretch budget.
Across 1.29 million nets, FABO reduces average FLUTE-normalized wirelength at every setting, with peak reductions of $0.83\%$ overall and $2.66\%$ for nets with at least 30 pins; its advantage grows with fanout, and its trees retain more shared root-side wire.
FABO-FAST makes this improvement practical by identifying and optimizing most high-gain nets within $1.3\times$ SALT's runtime.

% --- Bibliography ---
\bibliographystyle{IEEEtran}
\bibliography{refs}

% Liu Shang
\begin{IEEEbiography}[{\includegraphics[width=1in,height=1.25in, clip,keepaspectratio]{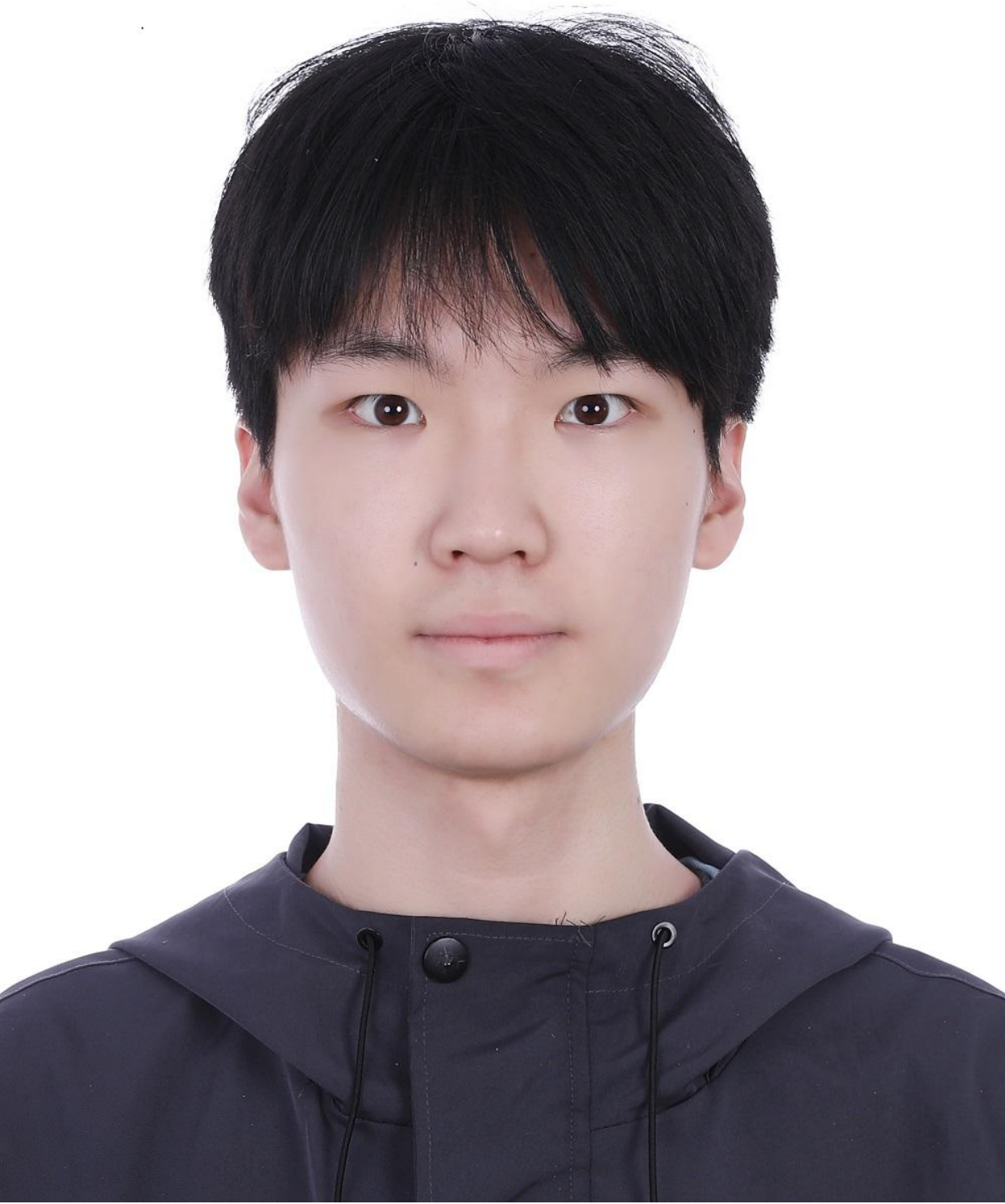}}]{Shang Liu}
received the B.E. degree in Automation Science and Electrical Engineering from Beihang University, Beijing, China, in 2023. He is currently pursuing the Ph.D. degree with the Department of Electronic and Computer Engineering, Hong Kong University of Science and Technology, Hong Kong. His research interests include agile VLSI design methodologies and Artificial Intelligence.
\end{IEEEbiography}

\begin{IEEEbiography}[{\includegraphics[width=1in,height=1.25in, clip,keepaspectratio]{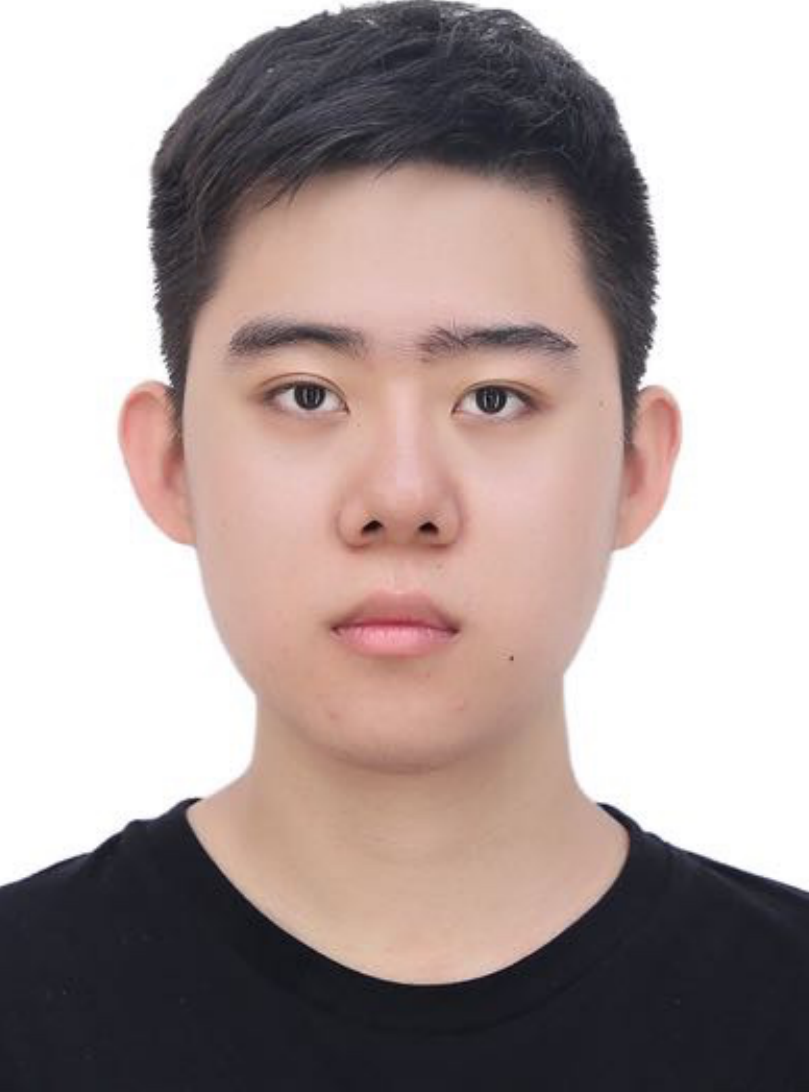}}]{Wenji Fang} is currently a Ph.D. student with the Department of Electronic and Computer Engineering at the Hong Kong University of Science and Technology. He received his M.Phil. degree in Microelectronics from the Hong Kong University of Science and Technology (Guangzhou) in 2024, and his B.Eng. degree from Nanjing University of Aeronautics and Astronautics in 2021. 
His research interests include Electronic Design Automation (EDA) and VLSI design verification.
\end{IEEEbiography}

\begin{IEEEbiography}[{\includegraphics[width=1in,height=1.25in, clip,keepaspectratio]{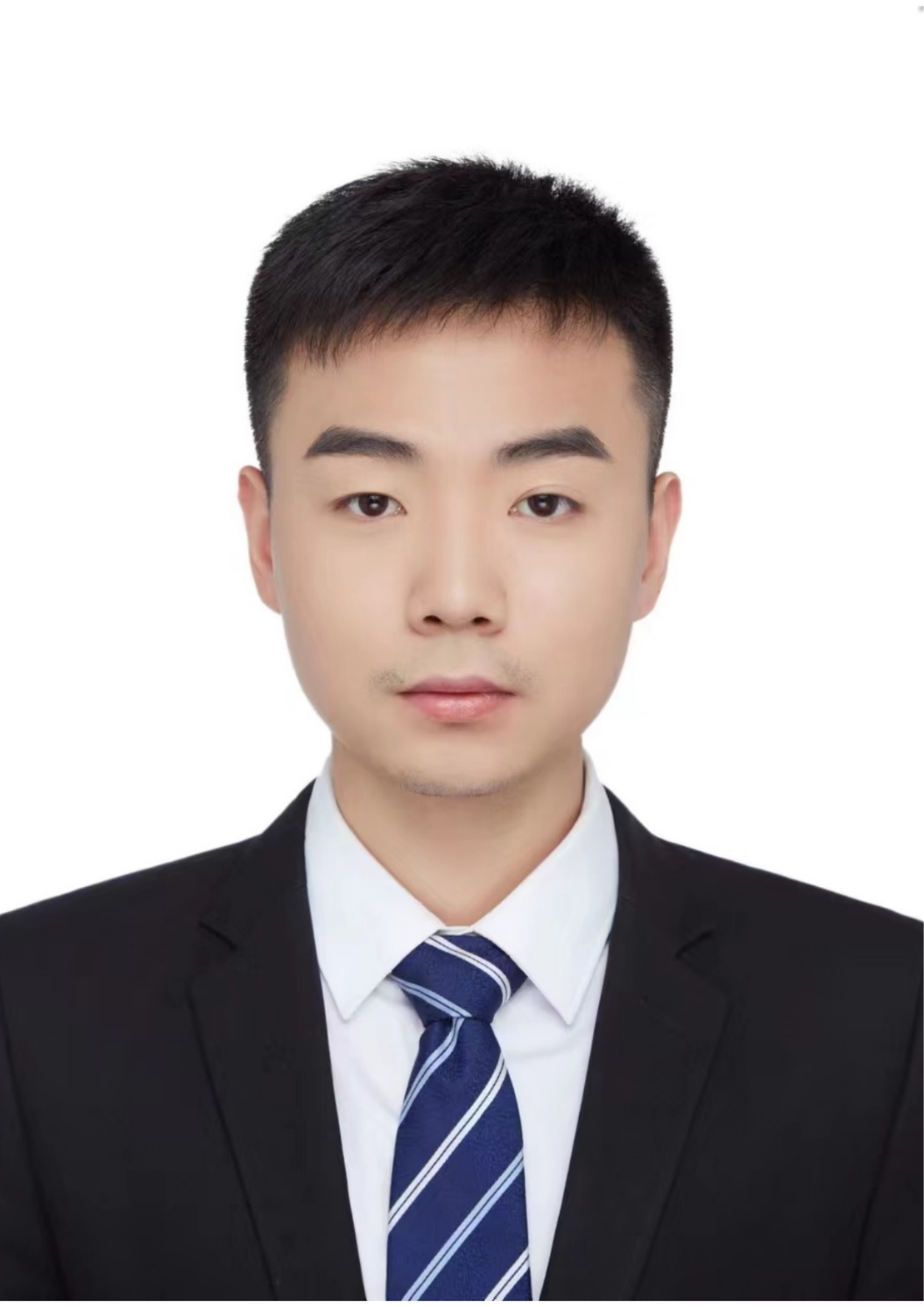}}]{Jing Wang}
received the B.S. degree in Electrical Information Engineering from Peking University, Beijing, China, in 2022, and the master's degree in Artificial Intelligence from the Department of Statistics and Actuarial Science, The University of Hong Kong, China, in 2023. He is currently pursuing the Ph.D. degree with the Department of Electronic and Computer Engineering, Hong Kong University of Science and Technology, Hong Kong. His research interests include agile VLSI design methodologies and Artificial Intelligence.
\end{IEEEbiography}

% \begin{IEEEbiography}[{\includegraphics[width=1in,height=1.25in, clip,keepaspectratio]{photo/qijun}}]{Qijun Zhang}
% received the B.Eng. degree from Tongji University, Shanghai, China, in 2022. He is currently a Ph.D. student in the Department of Electronic and Computer Engineering (ECE) at the Hong Kong University of Science and Technology (HKUST). His research interests include Computer Architecture and Electronics Design Automation.
% \end{IEEEbiography}

% \begin{IEEEbiography}[{\includegraphics[width=1in,height=1.25in,clip,keepaspectratio]{photo/hongcezhang}}]{Hongce Zhang}
% (Member, IEEE) received the B.S. degree
% in microelectronics from Shanghai Jiao Tong
% University, Shanghai, China, in 2015, and the
% Ph.D. degree from the Electrical and Computer
% Engineering Department of Princeton University, NJ, USA, in 2021.

% He is currently an Assistant Professor with the Microelectronics Thrust, Function Hub of Hong Kong University of Science and Technology (Guangzhou), Guangzhou, China, and is also affiliated with the Electronic and Computer Engineering Department of the Hong Kong University of Science and Technology, Clear Water Bay, Hong Kong SAR.
%  His research interests
% include formal verification and hardware model checking.
% \end{IEEEbiography}

\begin{IEEEbiography}[{\includegraphics[width=1in,height=1.25in,clip,keepaspectratio]{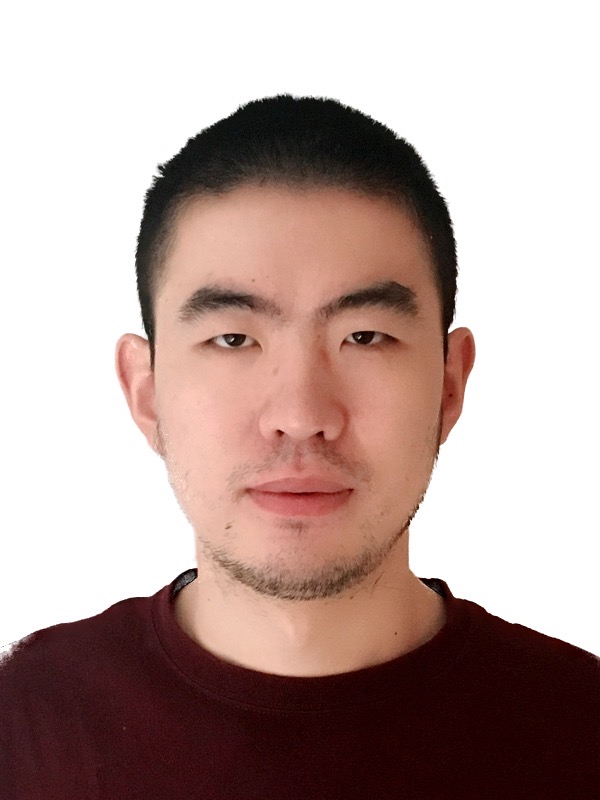}}]{Hongxin Kong} received the B.S. degree in electrical engineering from Nanjing University of Aeronautics and Astronautics, Nanjing, China, in 2016, and the Ph.D. degree in computer engineering from Texas A\&M University, College Station, TX, USA, in 2020. He is currently a senior staff R\&D engineer with Synopsys, Inc. Prior to joining Synopsys, he was with Advanced Micro Devices (AMD) and Cadence. His research interests include electronic design automation (EDA), GPU-accelerated physical design algorithms, place and route (P\&R), and static timing analysis.
\end{IEEEbiography}

% Lu Yao
\begin{IEEEbiography}[{\includegraphics[width=1in,height=1.25in, clip,keepaspectratio]{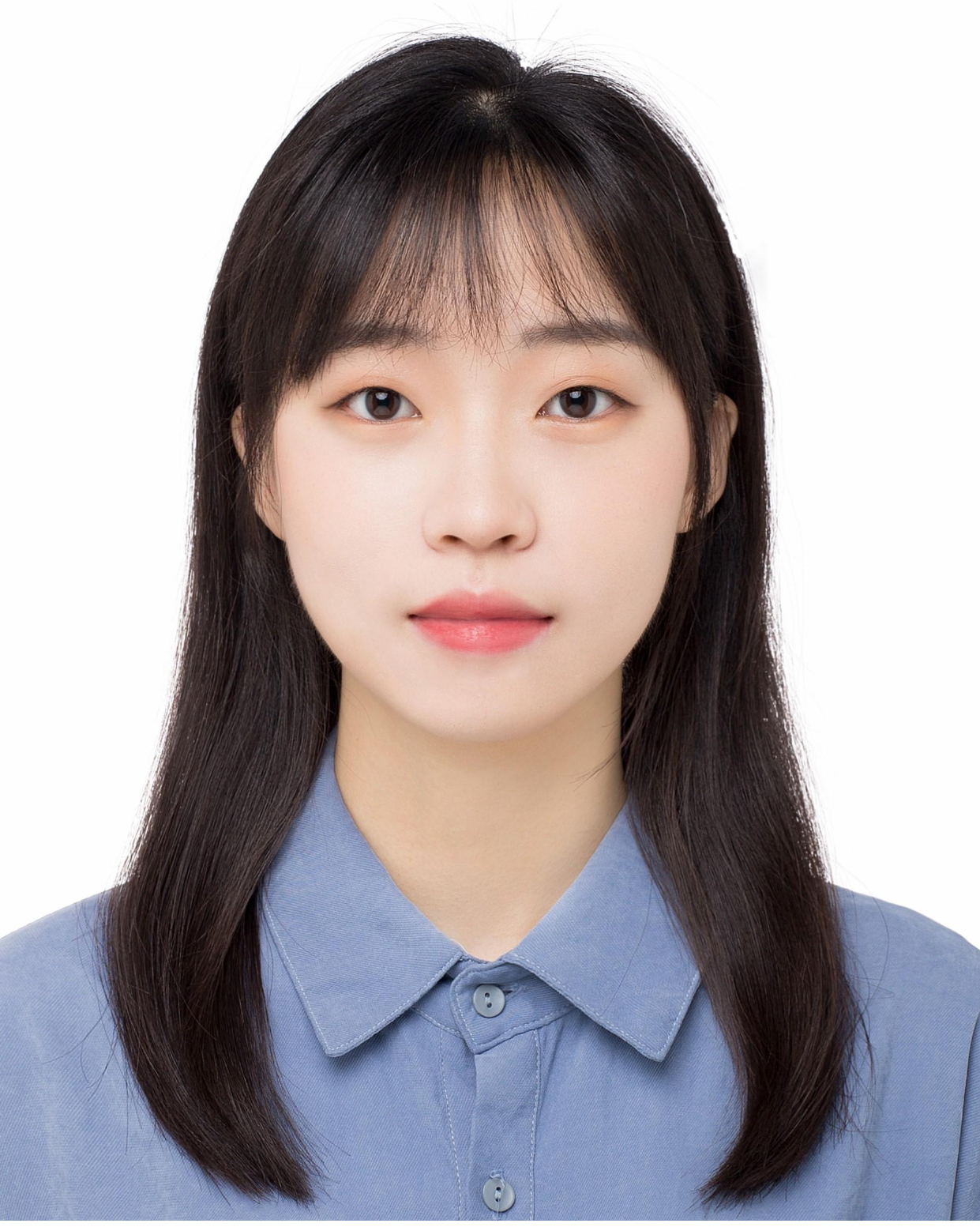}}]{Yao Lu}
received the B.E. degree from the School of Electronic Science and Engineering, Southeast University, Nanjing, China, in 2020, and the master's degree from the School of Microelectronics, Fudan University, Shanghai, China, in 2023. She is currently pursuing the Ph.D. degree with the Department of Electronic and Computer Engineering, The Hong Kong University of Science and Technology, Hong Kong. Her current research interests focus on machine learning applications in EDA.
\end{IEEEbiography}

\begin{IEEEbiography}[{\includegraphics[width=1in,height=1.25in,clip,keepaspectratio]{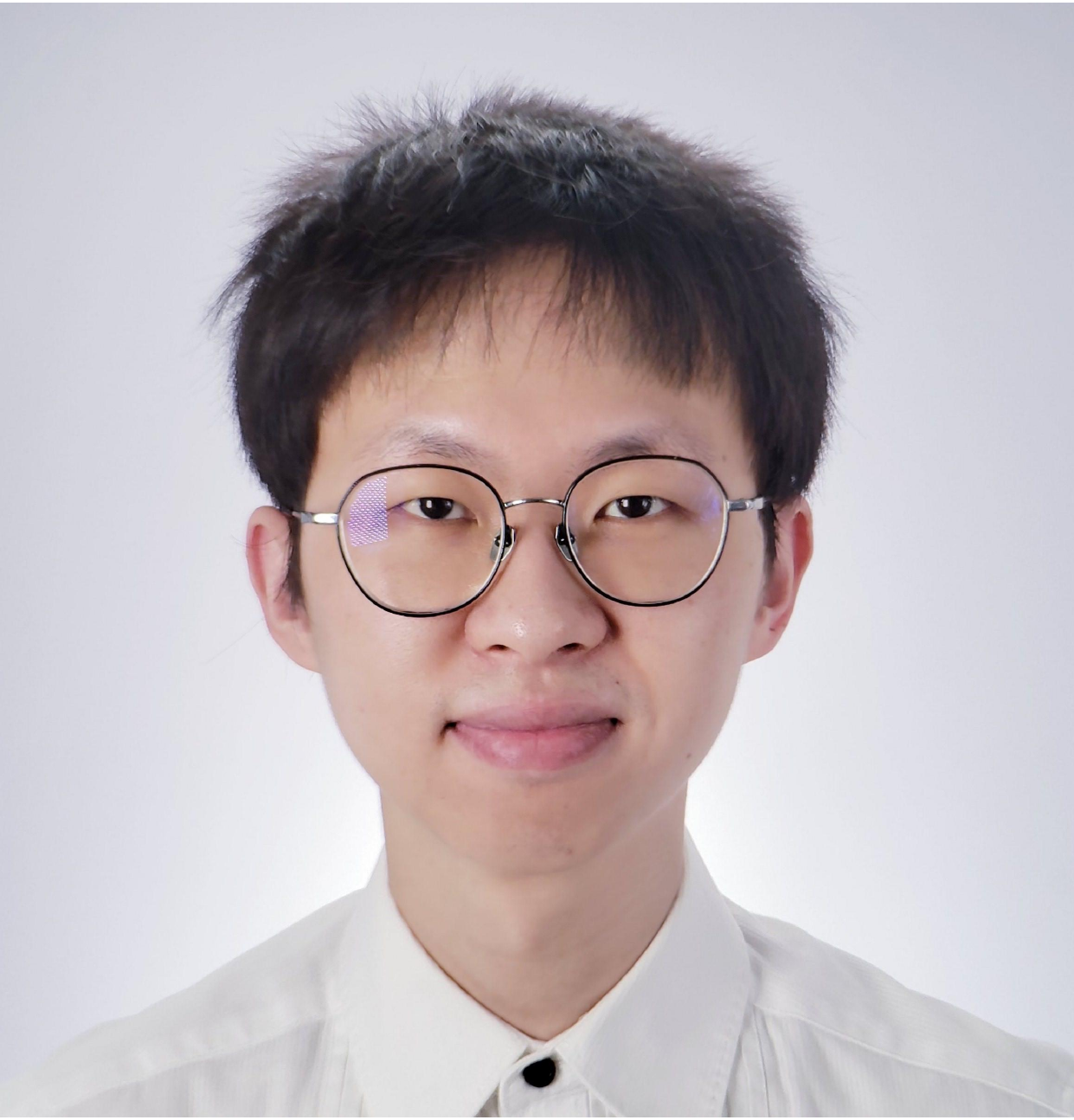}}]{Zhiyao Xie}
is an Assistant Professor in the Department of Electronic and Computer Engineering (ECE) at the Hong Kong University of Science and Technology (HKUST). He received his Ph.D. degree from Duke University in 2022 and the B.Eng. degree from City University of Hong Kong in 2017. His research focuses on AI-driven techniques for EDA and VLSI design. He has received multiple prestigious awards, including the ACM SIGDA Outstanding New Faculty Award 2026, ACM ASPLOS 2026 Best Paper Award, IEEE/ACM MICRO 2021 Best Paper Award, ASP-DAC 2023 Best Paper Award, ACM SIGDA SRF Best Research Poster Award 2022, ACM Outstanding Ph.D. Dissertation Award in EDA 2023, EDAA Outstanding Ph.D.  Dissertation Award 2023, and Hong Kong Research Grants Council (RGC) Early Career Award 2023.
\end{IEEEbiography}
\end{document}